\documentclass
{article}
\usepackage{amssymb, amsfonts,ifthen}

\usepackage{amsthm}
\newtheorem{theorem}{Theorem}
\newtheorem{lemma}{Lemma}
\newtheorem{proposition}{Proposition}
\newtheorem{corollary}{Corollary}

\theoremstyle{definition}

\newtheorem{definition}{Definition}
\newtheorem{remark}{Remark}

\newcommand{\qedi}{ }

\newcommand{\be}{\begin{equation}}
\newcommand{\ee}{\end{equation}}
\newcommand{\bea}{\begin{eqnarray}}
\newcommand{\eea}{\end{eqnarray}}
\newcommand{\bra}{\langle} \newcommand{\ket}{\rangle}

    \newcommand {\norm} [2] [] {\ensuremath{ \left\Vert  #2  \right\Vert_{#1} } }
    \newcommand {\R} {\ensuremath{\mathbb{R}}}
    
    \newcommand {\N} {\ensuremath{\mathbb{N}}}

    \newcommand {\om} {\omega}

    \newcommand {\mr} {\mathrm}

    \newcommand {\C} {\mathbb C}

  \newcommand{\commut}[2]{\left[ #1 , #2 \right]}
  \newcommand{\antcom}[2]{\left[ #1 , #2 \right]_+}
  \newcommand{\eps}{\varepsilon}
 \newcommand{\Or}{\mathcal{O}}
  \newcommand{\e}{\mathrm{e}}
  
  \newcommand{\tRe}{\mathrm{Re}}
  \newcommand{\tIm}{\mathrm{Im}}
\newcommand{\id}{\mathrm{id}}
\newcommand{\D}{\mathrm{d}}
\newcommand{\E}{\mathrm{e}}
\newcommand{\I}{\mathrm{i}}

\newcommand{\tc}{\xi_{\mathrm{c}}}
\newcommand{\tr}{\xi_{\mathrm{r}}}
\newcommand{\M}{M}

\newcommand{\ec}{\mathcal{E}}
\newcommand{\bk}[1]{\bra #1 \ket}
\newcommand{\xicrit}{\xi_{\mr{crit}}}

\newcommand{\facnorm}[2]{\norm[(#2)]{#1}}
\newcommand{\perturb}[1]
{
    \ifthenelse{\equal{#1}{\theta}}{\theta_{\mathrm r}}{}
    \ifthenelse{\equal{#1}{x}}{\xi}{}
    \ifthenelse{\equal{#1}{y}}{\eta}{}
    \ifthenelse{\equal{#1}{z}}{\zeta}{}
}

\newcommand{\rom}{\renewcommand{\labelenumi}{{\rm(\roman{enumi})}}}

\newcommand{\alp}{\renewcommand{\labelenumi}{{\rm(\alph{enumi})}}}

\begin{document}

\title{\LARGE \bf Emergence of exponentially small reflected waves}

\author{Volker Betz\thanks{Supported by an EPSRC fellowship EP/D07181X/1}\\{\normalsize
Mathematics Institute, University of
Warwick, United
Kingdom} \\{\normalsize betz@maths.warwick.ac.uk}
\\[5mm] 
Alain Joye
 \\{\normalsize  Institut Fourier,
Universit\'e de Grenoble I, BP 74, 38402 St.-Martin-d'H\`eres, France}\\{\normalsize alain.joye@ujf-grenoble.fr}
\\[5mm] 
Stefan Teufel\\{\normalsize
Mathematisches Institut, Universit\"at T\"ubingen, Germany}\\{\normalsize stefan.teufel@uni-tuebingen.de}}

\date{\today}

\maketitle

\begin{abstract}
We study the time-dependent scattering of a quantum mechanical wave packet at a barrier for 
energies larger than the barrier height, in the semi-classical regime. More precisely, we are 
interested in the leading order of the exponentially small scattered part of the wave packet in the semiclassical parameter when the energy density of the incident wave is sharply peaked around some value. We prove that this reflected part has, to leading order, 
a Gaussian shape centered on the classical trajectory for all times soon after 
its birth time. We give explicit formulas and rigorous error bounds for the reflected 
wave for all of these times. 

\medskip

  \noindent
  MSC (2000):
  \underline{41A60};
  81Q05.\\
  Key words: Above barrier scattering, exponential asymptotics, quantum theory.
\end{abstract}
 
\section{Introduction}

We consider the problem of potential scattering for a quantum particle in one dimension in the semiclassical limit. Let $V:\R\to \R$ be a bounded analytic potential function such that $\lim_{|x|\to\infty} V(x) =0$.  Then a classical particle approaching the potential with total energy smaller than $\max_{x\in\R} V(x)$ is completely  reflected.  It is well known that for a quantum particle tunneling can occur, i.e.\ parts of the wave function can penetrate the potential barrier. Similarly, for energies larger than the barrier height $\max_{x\in\R} V(x)$ no reflection occurs for classical particles, but parts of the wave function of a quantum particle might be reflected. In the semiclassical limit the non-classical tunneling probabilities respectively the non-classical  reflection probabilities are exponentially small in the semiclassical parameter. 
While the computation of the semiclassical tunneling probability can be found in standard textbooks like \cite{LaLi}, the problem of determining  the reflection coefficient in the case of above barrier scattering is much more difficult \cite{LaLi, FrFr}. In this paper we are not only interested in the scattering limit but also in the questions where, when and how the reflected piece of the wave function emerges in an above barrier scattering situation. 

The time-dependent Schr\"odinger equation for a quantum particle moving in one dimension in  the potential $V$ is 
\be\label{tdschINTRO}
i\eps \partial_t \Psi=(-\eps^2\partial_x^2+V(x))\Psi, \ \ \ 
\Psi(\cdot, t, \eps)\in L^2(\R),
\ee
where $0<\eps \ll 1$ is the semiclassical parameter. 
A solution to (\ref{tdschINTRO}) can  be given in the form of a wave packet: for any fixed energy $E>\max_{x\in\R}V(x)$,
let $\phi(x,E)$ be a solution to the stationary Schr\"odinger equation
\be\label{sschINTRO}
-\eps^2\partial_x^2\phi=(E-V(x))\,\phi=: p^2(x,E)\,\phi,
\ee
where 
\[
p(x,E)=\sqrt{E-V(x)} >0
\]
is the classical momentum at energy $E$. 
Then the function 
\be\label{superpINTRO}
\Psi(x,t,\eps):=\int Q(E,\varepsilon)e^{-itE/\varepsilon} 
\phi(x,E)\ dE,
\ee
for some regular enough energy density $Q(E,\eps)$ is a solution of (\ref{tdschINTRO}).
We will require boundary conditions on the solution of (\ref{sschINTRO}) that lead to 
a wave packet that  is incoming from $x=+\infty$ and we will assume the energy density is
sharply peaked around a value $E_0>\max_{x\in\R}V(x)$. 

In order to understand the emergence of exponentially small reflected waves from (\ref{superpINTRO}) we basically need to solve two problems. Roughly speaking, we need to first  determine the solutions to (\ref{sschINTRO}) with appropriate boundary conditions  up to errors sufficiently small as $\eps\to 0$. Since we are interested in the solution of   (\ref{tdschINTRO}) for all times, we need to determine the solution of  (\ref{sschINTRO}) and also need to split off the piece corresponding to a reflected wave  for all $x\in \R$ and not only asymptotically for $|x|\to\infty$. The large $|x|$s regime of (\ref{sschINTRO}) is rigorously studied in \cite{jp95,Ra} by means of complex WKB methods and yields the exponentially small leading order of the reflected piece. On the other hand, the behaviour for all $x$s of  (\ref{sschINTRO}) has been investigated at a theoretical physics level in \cite{Be1}, where error function behaviour of the reflected wave piece has first been derived. However, our method is quite different and allows for a rigourous treatment, which is one main new contribution of the present paper.
 
In a second step we have to evaluate (\ref{superpINTRO}) using the approximate $\phi$s from step one and extract the leading order expression for the reflected part of the wave, for a suitable energy density $Q(E,\varepsilon)$. 
Since we know that the object of interest, namely the reflected wave, is exponentially small in $\eps$, both steps have to be done with great care and only errors smaller than the exponentially small leading order quantity are allowed. This has so far been done only in the scattering regime for similar problems \cite{hj05,jm05}. By contrast, we will be able to carry out the analysis for all times just after the birth time of the reflected wave, not only in an asymptotic regime of large positions and large times. This is the second main new aspect of our work.

Before further explaining the result, let us give some more details on how to approach  the two problems mentioned. For solving (\ref{sschINTRO}) we use techniques developed in \cite{BeTe1,BeTe2}, cf.\ also \cite{hj04} in the context of adiabatic transition histories. The connection to adiabatic theory is easily seen by writing the second order ODE (\ref{sschINTRO}) as a system of first order ODEs.
Put
\[
\psi(x,E) = \left( \begin{array}{l} \phi(x,E) \\  \I \eps \phi'(x,E) \end{array} \right)\,,
\]
where the prime denotes one derivative with respect to $x$, then 
  $\phi$ solves (\ref{sschINTRO}) if and only if $\psi$ solves 
\be \label{2level2INTRO}
\I \eps \partial_{x} \psi (x,E)= \pmatrix{0 & 1 \cr p^2(x,E) & 0}\psi(x,E)=:
H(x,E)\psi(x,E)\,.
\ee
In \cite{BeTe2} we studied the solutions of the   adiabatic problem 
\be \label{2level2adiabatic}
\I \eps \partial_{t} \psi(t) =  \tilde H(t)\psi(t)\,,
\ee
where $\tilde H(t)$ is a time-dependent real symmetric $2\times 2$-matrix. So the two differences between (\ref{2level2INTRO}) and (\ref{2level2adiabatic}) are that $H(x,E)$ is not symmetric, but still has two real eigenvalues, and that  an additional parameter dependence on $E$ occurs.  In Section~2 we explain how the results from \cite{BeTe2} on (\ref{2level2adiabatic}) translate to (\ref{2level2INTRO}). 

From this  we find in Section~3 that the solution to  (\ref{sschINTRO})  can be written as
\be\label{SPLIT}
\phi(x,E) = \phi_{\rm left}(x,E) + \phi_{\rm right}(x,E)\,,
\ee
where $ \phi_{\rm left}(x,E)$ corresponds to a wave traveling to the left and is supported on all of $\R$, while $\phi_{\rm right}(x,E)$ is the reflected part   traveling to the right and is essentially supported to the right of the potential. The difference between (\ref{SPLIT}) and the usual WKB splitting   
\be\label{WKBsplit}
\phi(x,E) = f_{\rm left}(x,E) \frac{\E^{\frac{\I}{\eps}\int_0^x p(r,E)\D r} }{\sqrt{p(x,E)}} + f_{\rm right}(x,E) \frac{\E^{-\frac{\I}{\eps}\int_0^x p(r,E)\D r} }{\sqrt{p(x,E)}} 
\ee
is that  the reflected part $\phi_{\rm right}(x,E)$ is exponentially small for all $x\in \R$ and not only for $x\to +\infty$ as $f_{\rm right}(x,E)$. Indeed, we will show that $\phi_{\rm right}(x,E)$ is of order 
$\e^{-\xi_{\mr c}(E)/\eps}$ for an energy dependent decay rate $\xi_{\mr c}(E)$. In the scattering limit $|x|\to\infty$ the splittings (\ref{SPLIT}) and (\ref{WKBsplit}) agree; so the main point about (\ref{SPLIT}) is that we can speak about the exponentially small reflected wave also at finite values of $x$. 

Inserting (\ref{SPLIT}) into (\ref{superpINTRO}) allows us to also split  the time-dependent  wave packet
\be\label{superpSPLIT}
\Psi(x,t)=\int Q(E,\varepsilon)e^{-itE/\varepsilon} \big(
\phi_{\rm left}(x,E) + \phi_{\rm right}(x,E)\big)dE =: \Psi_{\rm left}(x,t)+\Psi_{\rm right}(x,t)\,.
\ee
Since $\phi_{\rm right}(x,E)$ is essentially of order $\E^{-\xi_c(E)/\eps}$,  the leading order contribution to $\Psi_{\rm right}(x,t)$ comes from energies $E$ near the maximum $E^*(\eps)$ of  $|Q(E,\eps)|\E^{-\xi_c(E)/\eps}$. In order to see nice asymptotics in the semiclassical limit $\eps\to 0$ we impose conditions on  $Q(E,\eps)$ that guarantee  that $E^*(\eps)$ has a limit $E^*$ as $\eps\to 0$. Typically, this will be true for a density $Q(E,\eps)$ which is sharply peaked at some value $E_0$ as $\eps\rightarrow 0$, in accordance with the traditional picture of a semiclassical wave packet. In particular, it should be noted that in general, the critical energy $E^*$ is different from $E_0$, the energy on which the density concentrates.

The main results of our paper  are increasingly explicit  formulas for the leading order exponentially small reflected wave $\Psi_{\rm right}(x,t)$ not only in the scattering limit, covered by \cite{hj05,jm05}, but for most finite times and positions. 
In particular these formulas show  where, when and how the non-classical reflected wave 
emerges. Let us note that similar questions can be asked for more general dispersive 
evolution equations, in the same spirit as the systems considered in \cite{jm05}. However, 
the missing piece of information that forbids us to deal with such systems is an equivalent 
of the analysis performed in \cite{BeTe1,BeTe2}, and adapted here to the scattering setup, 
that yields the exponentially small leading order of the solutions to (\ref{sschINTRO}), 
for all $x$s.

Our methods and results belong to the realm of  semiclassical analysis, in particular, to exponential asymptotics and the specific problem we study has been considered in numerous works. Rather than attempting to review the whole litterature relevant to the topic, we want to point out the differences with respect to the results directly related to the problem at hand.  Exponentially small bounds can be obtained quite generally for the type of problem we are dealing with, see e.g. \cite{Fe0, Fe, JKP, Sj, JoPf1, Ne, mz, hj}... Moreover, existing results on the exponentially small leading order usually are either obtained for the scattering regime only, \cite{Fe0, Fe,jp95,Ra,Jo,hj05,jm05}, or non-rigorously \cite{Be1,Be2,BeLi}, or both \cite{WiMo}. To our knowledge, the only exceptions so far for the time independent case are \cite{BeTe1,BeTe2,hj04}, which we follow and adapt, but as mentioned above we need to be even more careful here in order to get uniform error terms in the energy variable. Finally, we focus on getting explicit formulae for the leading term of the reflected wave, rather than proving structural theorems and general statements, which justifies our choice of energy density.

The detailed statements of our results are too involved to sum them up in the introduction; let us instead mention a few qualitative features. In the semiclassical limit $\eps\to 0$ the reflected part of the wave is localized in a $\sqrt{\eps}$-neighborhood of a classical trajectory $q_t$ starting at a well defined transition time at a well defined  position $q^*$ with  velocity  $\sqrt{E^*-V(q^*)}$. 
A priori we distinguish three regimes. A $\sqrt{\eps}$-neighborhood of the transition  point $q^*$, which  is located where  a certain Stokes line in the complex plane crosses the real axis, is the birth region, where the reflected part of the wave emerges. It remains $\sqrt{\eps}$-localized near the trajectory $q_t$ of a classical particle in the potential $V$ with energy $E^*>$ max$V$ for all finite times, which defines our second regime. Therefore $\lim_{t\to\infty}q_t = \infty$, i.e.\ the trajectory belongs to a scattering state and moves to spatial infinity. The third regime is thus the scattering region, where the wave packet moves freely. We give precise characterizations of the reflected wave in all three regimes, Theorem~\ref{reflected wave shape}. It turns out that as soon as the reflected wave leaves the birth region, it is well approximated by a Gaussian wave packet centered at the classical trajectory $q_t$, see Theorem~\ref{chimod gauss}.

We now give a plan of the paper. In order to precisely formulate our main 
results, we need to discuss in detail under which conditions and in which sense we can solve  (\ref{2level2INTRO}) up to exponentially small errors. This is discussed in  Section~2 where the main result is Theorem~\ref{solution}. The proofs are based on the techniques developed in \cite{BeTe1,BeTe2, BeTe3} and can be found in Sections~5~and~6. In Section~3 we explain how to arrive   from our solutions of (\ref{2level2INTRO}) at the main theorems, Theorems~\ref{chi eff thm}--\ref{chimod gauss}, by first obtaining solutions of (\ref{sschINTRO}) split according to (\ref{SPLIT}) and then evaluating (\ref{superpSPLIT}). The mathematical details on the asymptotic analysis of (\ref{superpSPLIT}) are covered in Section~4, where the main tool is Laplace's method. In Section~5 we explain how to bring  (\ref{2level2INTRO})   into an almost diagonal form, with off-diagonal elements   of order $\eps^n$ for arbitrary $n\in\N$, see Theorem~\ref{general diag}. We proceed analogous  to the well known transformation to super-adiabatic representations in adiabatic theory, see \cite{BeTe1,BeTe2} for references. 
In Section~6 we analyze the asymptotic behavior of the off-diagonal terms in the super-adiabatic representations. By optimal truncation we find the optimal super-adiabatic representation in which the off-diagonal terms are exponentially small and explicit at leading order. Then standard perturbation theory yields Theorem~\ref{solution}.

\section{Superadiabatic representations for non-selfadjoint Hamiltonans} \label{general theorems}

In this section we adapt the results of \cite{BeTe1} and \cite{BeTe2}
to fit our needs. The main difference is that the Hamiltonian matrix we will have to deal with is not self-adjoint; this will require some adjustments which we carry out in the Section \ref{asympt recursion}, but fortunately much of the theory from \cite{BeTe1, BeTe2} carries over. 
The equations we consider are 
\begin{equation} \label{ODE}
\I \eps \partial_{x} \psi(\eps,x,E) =  H(x,E) \psi(\eps,x,E) := \left( \begin{array}{cc} 0 & 1 \\ p^{2}(x,E) & 0 \end{array}\right) \psi(\eps,x,E)
\end{equation}
where $p(x,E) = \sqrt{E-V(x)}$, and $V$ is analytic in a neighbourhood of the real axis with $\sup_{x \in \R} V(x) < E$. 
We will here give the optimal superadiabatic representation for this system, which is an $\eps$-dependent basis transform that diagonalizes (\ref{ODE}) up to exponentially small errors while still agreeing with the eigenbasis of $H(x,E)$ for $|x| \to \infty$. In this basis, for suitable initial conditions the second component of the solution is given by an exponentially small error function with superimposed oscillations, to leading order; this fact will enable us to describe the time development of the exponentially small reflected wave. 

We start by introducing the invertible transformation 
\begin{equation} \label{xinat}
\xi(x,E) =  2 \int_{0}^{x} p(y,E) \, \D y.
\end{equation}
Equation (\ref{ODE}) then transforms into 
\begin{equation} \label{ODE2}
\I \eps \partial_{\xi} \psi(\eps,\xi,E) =  H(\xi,E) \psi(\eps,\xi,E) := \frac{1}{2}\left( \begin{array}{cc} 0 & \frac{1}{\tilde p(\xi,E)}\\ \tilde p(\xi,E) & 0 \end{array}\right) \psi(\eps,\xi,E),
\end{equation}
where $\tilde p(\xi(x,E),E) = p(x,E)$. 
Let us for the moment consider fixed $E$ and drop the dependence on $E$ from the notation. The matrix on the right hand side of (\ref{ODE2}) has eigenvalues $-1/2$ and $1/2$, and is diagonalized by the basis transform implemented through
\[
T(\xi) = \left( \begin{array}{cc} \sqrt{\tilde p(\xi)} & 1/\sqrt{\tilde p(\xi)} \\ \sqrt{\tilde p(\xi)} & -1/\sqrt{\tilde p(\xi)} \end{array}\right).
\]
We then find 
\begin{equation} \label{adiabODE} 
0 = T \left( \I \eps \partial_{\xi} - H(\xi) \right) T^{-1} (T \psi) = \left( \I \eps \partial_{\xi} - \frac{1}{2}\left( \begin{array}{cc} 1 & \I \eps \theta'(\xi) \\ \I \eps \theta'(\xi) & -1 \end{array}\right) \right) \psi_{\rm a},
\end{equation}
where $\theta'(\xi) = \frac{\tilde p'(\xi)}{\tilde p(\xi)}$ is the adiabatic coupling function. The following theorem gives the $n$-th superadiabatic representation for every $n \in \N$.

\begin{theorem} \label{n-th}
Assume $\theta'(\xi) \in C^{\infty}(\R)$. Then for each $n \in \N_0$ there exists an invertible matrix $T_{n} = T_{n}(\eps, \xi)$ such that 
\begin{equation} \label{superadODE} 
T_{n} \left( \I \eps \partial_{\xi} - H \right) T_{n}^{-1} = 
 \I \eps \partial_{\xi} - \left( \begin{array}{cc} \rho_{n+1}(\eps,\xi) & \eps^{n+1} k_{n+1}(\eps,\xi) \\[2mm] - \eps^{n+1} \overline{{k}_{n+1}}(\eps,\xi) & - \rho_{n+1}(\eps,\xi) \end{array}\right) 
\end{equation}
where $\rho_{n}(\eps,\xi) = \frac{1}{2} + \Or(\eps^{2})$, 
$k_{n}(\eps,\xi) = (x_{n}(\xi) - z_{n}(\xi)) (1+\Or(\eps))$,
and the $x_{j}$, $z_{j}$ can be calculated from the recursive equations
\begin{eqnarray}
 x_{1} &=& \frac{\I}{2} \theta', \qquad \qquad y_{1}=z_{1}=0, \label{start}\\
 x_{n+1} & = &  -\I z_n' +\I \theta' y_n\label{xRec1}\\
 y_{n+1} & = &  \sum_{j=1}^n \left( x_j x_{n+1-j} +y_j y_{n+1-j}    -z_j z_{n+1-j}   \right)\label{yRec1}\\
 z_{n+1} & = &  -\I x_n'\,.\label{zRec1}
\end{eqnarray}
\end{theorem}
We will give the proof of this theorem in Section \ref{recursion}. There we will even give explicit expressions of $T_{n}$, $\rho_{n}$ and $k_{n}$ in terms of the $x_{i}$, $y_{i}$ and $z_{i}$, but these are too messy to write down here. Note that none of them is a polynomial in $\eps$.

As it stands (\ref{superadODE}) is not useful, since neither the $n$-dependence nor the $\xi$-dependence of the off-diagonal terms is controlled. 
In order to achieve this, we need to make stronger assumptions on $\theta'$. As the present situation is very close to the superadiabatic transition histories considered in \cite{BeTe2}, the conditions will be identical: we require that $\theta'$ is analytic in a strip around the real line, and that the singularities which limit this strip are first order poles plus lower order corrections. As explained in \cite{BeTe2}, those corrections are controlled using the following norm on functions of a real variable:

\begin{definition}\label{normDef}
Let  $\tc > 0$, $\alpha > 0$ and  $I \subset \R$ be an interval. For
$f \in C^\infty(I)$ we define
\begin{equation} \label{norm}
\facnorm{f}{I,\alpha,\tc} := \sup_{t \in I} \sup_{k \geq 0} \left|
\partial^k f(t) \right| \frac{\tc^{\alpha + k}}{\Gamma(\alpha + k)}
\leq \infty
\end{equation}
and
$$ F_{\alpha,\tc}(I) = \left \{ f \in C^{\infty}(I): \facnorm{f}{I,\alpha,\tc} < \infty \right\}.$$
\end{definition}

We now state the assumptions that we will make on the adiabatic coupling function. The first one will be required for regions of $\R$ that are far away from the first order poles that are nearest to the real axis, while the second one will be required close to such poles.
The quantities $\tc$ and $\tr$ will be the imaginary and real part of those poles.\\[2mm]

\noindent{\bf Assumption 1:} {\it For a compact interval $I$ and
$\kappa\geq\tc>0$ let
$\theta'(\xi) \in F_{1,\kappa}(I)$.}\\

\noindent{\bf Assumption 2:} {\it  For $\gamma$, $\tr$, $\tc\in \R$
let
$$ \theta'_0(\xi) = \,\gamma\left(\frac{1}{\xi - \tr + \I \tc} + \frac{1}{\xi -\tr- \I \tc}\right)$$
be the sum of two complex conjugate first order poles located at
$\tr\pm\I\tc$ with residues $\gamma$. On a compact interval
$I\subset [\tr-\tc,\tr+\tc]$ with $ \tr \in I$ we assume that
 \begin{equation}
\theta'(\xi) = \theta'_0(\xi) + \theta_{\rm r}'(\xi)\quad\mbox{\it
with}\quad \theta_{\rm r}'(\xi)\in F_{\alpha,\tc}(I)
 \end{equation}
for some  $\gamma$, $\tc,\tr\in\R$, $0<\alpha<1$.}\\

Under these assumptions, the following theorem holds:

\begin{theorem}\label{MainTh}
Let $J_{\tc}\subset(0,\infty)$, $J_\alpha\subset(0,1)$ and
$J_\gamma\subset(0,\infty)$ be compact intervals, and for given 
$\tc>0$ let 
$0\leq \sigma_\eps <2$ be such that
\begin{equation}\label{ndef}
n_\eps = \frac{\tc}{\eps} -1+ \sigma_\eps\qquad\mbox{is an
even integer.}
\end{equation}

\begin{enumerate}\rom
\item There exists $\eps_0>0$ and a locally bounded function  $\phi_1:\R^+\to\R^+$ with $\phi_1(x)=\Or(x)$ as $x\to 0$,
such that the following holds: uniformly all in $\eps\in (0,\eps_0]$, all compact intervals $I \subset \R$ and all matrices $H(\xi)$ as given in
(\ref{adiabODE}) and satisfying Assumption~1 for some $\tc\in J_{\tc}$, the elements of the superadiabatic Hamiltonian as in (\ref{superadODE}) and the transformation
$T_{n_{\eps}} = T_{n_\eps}(\eps,\xi)$ with $n_\eps$ as in (\ref{ndef}) satisfy
\begin{eqnarray}
\left| \rho_{n_\eps}(\eps,\xi) - \frac{1}{2} \right| &\leq& \eps^2
\phi_1\!\left( \facnorm{\theta'}{I,1,\kappa}\right) \label{rhoc1}\\
\left|\eps^{n_\eps+1}k_{n_\eps}(\eps,\xi)\right| &\leq&
\sqrt{\eps}\,\E^{-\frac{\tc}{\eps}(1+\ln\frac{\kappa}{\tc})}
\phi_1\!\left(\facnorm{\theta'}{I,1,\kappa}\right)\label{rhoc2}
\end{eqnarray}
and
\begin{equation}\label{UminusUnull}
\|T_{n_\eps} - T \| \leq \eps
\phi_1\!\left(\facnorm{\theta'}{I,1,\kappa}\right).
\end{equation}\vspace{1mm}

\item Define
\be \label{gdef} 
g(\eps,\xi) = 2\I\,\sqrt{{\textstyle\frac{2\eps}{\pi
\tc}}}\,\sin\left({\textstyle\frac{\pi \gamma}{2}}\right)\,
\E^{-\frac{\tc}{\eps}}\,\E^{-\frac{(\xi-\tr)^2}{2\eps \tc}} \,\cos\left({\textstyle\frac{\xi-\tr}{\eps}
-\frac{(\xi-\tr)^3}{3\eps\tc^2} + \frac{\sigma_\eps (\xi - \tr)}{\tc}} \right).
\ee
There exists   $\eps_0>0$ and a locally bounded function  $\phi_2:\R^+\to\R^+$ with $\phi_2(x)=\Or(x)$ as $x\to 0$, such that the following holds: uniformly in all $\eps < \eps_{0}$, all compact intervals $I \subset \R$, and all $H(\xi)$ as in
(\ref{adiabODE}) satisfying Assumption~2 for some $\tc\in J_{\tc}$,
$\alpha\in J_\alpha$, $\gamma\in J_\gamma$, we have
\begin{equation}\label{Hod}
\left| \eps^{n_\eps+1}k_{n_\eps}(\eps,\xi) - g(\eps,\xi) \right| \leq
\eps^{\frac{3}{2}-\alpha}\E^{-\frac{\tc}{\eps}} \phi_2(M),
\end{equation}
where $M= \max \left\{ \facnorm{\theta'}{I,1,\tc}, \facnorm{\perturb \theta'}{I,\alpha,\tc} \right\}$.
Furthermore, (\ref{rhoc1}) and (\ref{UminusUnull}) hold with $\kappa=\tc$.
\end{enumerate}
\end{theorem}

In words, Theorem \ref{MainTh} states that we can find a transformation $T_{n}$ which is $\eps$-close to the one that diagonalizes $H(\xi)$, but which diagonalizes the full equation 
(\ref{adiabODE}) up to exponentially small errors. What is more, for those $\xi$ where the off-diagonal elements are close to maximal, they are given by explicitly known Gaussian functions with superimposed oscillations to leading order. The dependence on $E$, which we dropped
earlier from our notation, comes in only through $\theta'$, which means through $\tc$, $\tr$ and $\gamma$ in the leading term, and through $\facnorm{\theta'}{I,1,\tc}$ in the corrections. Since later we will have to consider a whole range of values for $E$, it is 
of crucial importance that all our error terms are small uniformly in basically all parameters, as stated above in Theorem \ref{MainTh}. On the other hand note that in those regions that are more than $\sqrt{\eps}$ away from $\tr$, the error term dominates in (\ref{Hod}). We will give a proof of those parts of Theorem \ref{MainTh} that are not identical to the situation in \cite{BeTe1,BeTe2} in Section  \ref{asympt recursion}.

Theorem \ref{MainTh} shows how to diagonalize equation (\ref{ODE2}) up to exponentially small errors. We will actually need an exponentially accurate solution to this equation in the whole $\xi$ domain. Here we will give a simple version of the corresponding result which is analogous to the one in \cite{BeTe3}, and for which only mild  integrability conditions on $\theta'$ and its derivatives are necessary. In \cite{BeTe3} it is shown that a sufficient condition is the integrability of $\xi \mapsto \facnorm{\theta'}{\xi,1,\tc+\delta}$ for some $\delta>0$, outside of a compact interval, and that a sufficient condition for this is the integrability of 
\[
\xi \mapsto \sup \{ |\theta'(z)|: z \in \C, |\xi-z| \leq \tc+\delta \}
\]
outside of a compact interval.

\begin{theorem} \label{solution}
Assume that $\theta'$ has two simple poles at $\pm \I \tc$ at leading order, and no further singularity of distance less or equal to $\tc$ to the real axis. In addition assume that for some $\delta > 0$, $\xi \mapsto \facnorm{\theta'}{\xi,1,\tc+\delta}$ is integrable outside of a compact interval. 
Let $\psi_{n}(\eps,\xi) = (\psi_{+,n}(\eps,\xi), \psi_{-,n}(\eps,\xi))$ be the solution of the equation (\ref{adiabODE}) in the optimal superadiabatic representation, i.e.\  
\begin{equation} \label{diffeq1}
T_{n_{\eps}} \left( \I \eps \partial_{\xi} - H \right) T_{n_{\eps}}^{-1}  \psi_{n}(\eps,\xi) = 0,
\end{equation}
subject to the boundary conditions
\[
\lim_{\xi \to -\infty} \psi_{+,n}(\eps,\xi) \E^{\frac{\I \xi}{2 \eps}} = 1, \quad \lim_{\xi \to -\infty} \psi_{-,n}(\eps,\xi) = 0.
\]
Then 
\[
\psi_{-,n}(\eps,\xi) = 2 
\sin\left(\frac{\pi\gamma}{2}\right)
\E^{-\frac{\tc}{\eps}}\E^{\frac{\I\xi}{2\eps} - \frac{\I \tr}{\eps}} {\rm
erf}\left(\frac{\xi - \tr}{\sqrt{2\eps\tc}}\right)
+\Or(\eps^{\frac{1}{2}-\alpha}\E^{-\frac{\tc}{\eps}})\,,
\]
where 
\[ 
{\rm{erf}}(x) = \frac{1}{\sqrt{\pi}} \int_{-\infty}^{x} \e^{-s^{2}} \, \D s
\]
is the error function normalized to be $0$ at $-\infty$.  
\end{theorem}

The proof of Theorem \ref{solution} requires no changes as compared to the one given in \cite{BeTe3} and is therefore omitted. Let us remark already here, however, that for our applications to the reflected wave we will need precise control over the $\xi$-dependence of the error term above, which means we will basically have to reprove a refined version of Theorem \ref{solution} later. Also for later purposes, we note that Theorem \ref{solution} also gives exponentially accurate information about $\psi_{+,n}$: defining 
\[
\psi_{+,\mr{eff}}(\xi) = \e^{-\frac{\I}{\eps} \int_{-\infty}^{\xi} \rho_{n}(\eps,r) \, \D r},
\]
we find using (\ref{diffeq1}) (given explicitly in (\ref{superadODE2})) that  
\begin{equation} \label{upper}
\left| \psi_{+,n}(\xi) - \psi_{+,\mr{eff}}(\xi) \right| = \left| \frac{1}{\eps} \int_{-\infty}^{\xi} \eps^{n+1} k_{n_{\eps}}(\eps,y) \psi_{-,n}(\eps,y) \, \D y\right| = \Or\left(\eps^{-1/2} \e^{-\frac{2 \tc}{\eps}}\right).
\end{equation}
The last equality follows from  (\ref{rhoc2}) and (\ref{gdef}) and the integrability condition in Theorem \ref{solution}. 
We end this section by linking the conditions on $\theta'$ to conditions on the function $p(x,E)$ appearing in (\ref{ODE}). The first condition 
is\\[3mm]

{\bf(A1)} We assume that $p(x,E) > c$ for some $c>0$ and all $x \in \R$, and that $p(\cdot, E)$ is analytic in a neighbourhood $U$ of the real axis. \\[3mm]

While the map $x \mapsto p(x,E)$ is
analytic on $U$ by (A1), it may have singular points in the complex plane. 
For each such point $z_0$ we define a straight half line $B_{z_0}$ pointing away from the real axis. In the case that a branch cut of $p$ is needed at $z_{0}$ we take it to lie on $B_{z_{0}}$. Since $U_0 := \C \setminus \bigcup_{z_0} B_{z_0}$ is simply connected, the function 
\be \label{nat scale}
\xi(z,E) = 2 \int_{0}^{z} p(y,E) \, \D y,
\ee
for $z \in U_{0}$ is single-valued if we require that the path of integration has to lie entirely in $U_{0}$. $\xi(z,E)$ is the unique analytic extension of $\xi$ as given in (\ref{xinat}) to $U_0$. For 
$x \in \R$ we put  
\[
C_r(x) := \{ z \in U_0: |\xi(z)-\xi(x)| < r \},
\]
denote by $C_{r,0}(x)$ the connected component of $C_r(x)$ containing $x$ and put
\[
R(x) := \sup \{r > 0: \overline{C_{r,0}(x)} \subset U_0 \},
\]
and $R_{0} = \inf_{x \in \R} R(x)$. 
Finally we put 
\be \label{U00}
U_{00} = \bigcup_{x \in \R} C_{R_{0},0}(x). 
\ee
The point of this construction is that now the set $\partial U_{00} \setminus U_{0}$ contains only those points on the boundary of $U_{0}$ that are closest to the real axis in the metric given by 
\be \label{metric}
d(z,z') = |\xi(z)-\xi(z')|.
\ee
It turns out that these are the points on which we need to impose conditions.\\[3mm]

{\bf(A2)} We assume that $\partial U_{00} \setminus U_{0} = \{ z_{\mr{crit}}, \overline z_{\mr{crit}} \}$ with $\tIm z_{\mr{crit}} > 0$, and that for $z \in U_{00}$ near $z_{\mr{crit}}$ we have 
\[ p(z,E) = (z- z_{\mr{crit}}(E))^{\beta(E)}(A(E) + g(z-z_{\mr{crit}},E)) \qquad \mbox{ for some  }\beta(E) > -1, \beta(E) \neq 0,\]
where $g$ is analytic near $z=0$ with $g(0,E) = 0$ and $A(E) \neq 0$. Furthermore, we assume that there exists $\delta > 0$ such that for all other critical points $z_{0}$ of $p(z,E)$ we have 
\be \label{safe distance}
\inf_{x \in \R} d(x,z_{0}) > |\tIm z_{\mr{crit}}| + \delta
\ee

For simplicity we assumed that $\partial U_{00} \setminus U_{0}$ just contains two complex conjugate points, but the extension to finitely many of them is straightforward in principle, though some of the estimates of Section \ref{s4} would get more messy. Moreover, slightly more general perturbations $g$ could be allowed, but we do not consider it worth the effort to include them. The conditions on $\beta$ ensure that $z_{\mr{crit}}$ is a critical point of $p$ such that $\lim_{z \to z_{\mr{crit}}} \xi(z)$ exists. Finally, (\ref{safe distance}) ensures that there is no sequence of singular points converging to the boundary of $U_{00}$ and is automatically fulfilled if e.g. $p$ has only finitely many singular points.

The integrability condition we need on $p$ can also be expressed in terms of the metric $d$ given in (\ref{metric}).
Let us define $h: \R \times \Delta \to [0,\infty]$ by 
\[
h(x,E) = \sup \left\{ \left| \frac{\partial_{x} p(y,E)}{p(y,E)^{2}} \right|: y \in C_{R_{0} + \delta}(x) \right\},
\]
where $\Delta\subset \R$ is the relevant set of energies, i.e.\ the support of the function $Q(E,\eps)$.
We assume: 
\\[3mm]

{\bf(A3)} There exists $\delta > 0$ such that
\[
\int_{\R \setminus [a,b]} h(x,E) p(x,E) \, \D x < \infty
\]
for some $a<b \in \R$.

\begin{proposition} \label{omega conditions}
Fix $E \in \R$. Assume that $p(.,E)$ fulfills (A1). For $x \in \R$ define 
$\xi(x)$ by (\ref{xinat}), and define $\tilde p$ through $\tilde p(\xi(x,E),E) = p(x,E)$. Put $\theta'(\xi,E) = \frac{\tilde p'(\xi,E)}{\tilde p(\xi,E)}$. 
\begin{itemize}
\item[(i)] Assume that $p(.,E)$ fulfils (A2). Then $\theta'$ fulfills Assumption 1 above on each compact subinterval of $\R$ containing $\tr$, with $\gamma = \frac{\beta}{\beta+1}$, 
\[
\tr(E) = \tRe(\xicrit), \quad \tc(E) = \tIm(\xicrit) \quad \mbox{where } \xicrit(E) := \xi(z_{\mr{crit}}(E),E),
\]
and for arbitrary $\alpha > 0$.\\
Moreover, $\theta'$ fulfils Assumption 2 on each compact subinterval of $\R$ such that $\inf\{|x - \tr|: x \in I\} > \delta$, where $\delta$ is the same as in (A2). In this case, $\kappa = \sqrt{\delta^{2} + \tc^{2}}$ in Assumption 2. 
\item[(ii)] If $p(.,E)$ fulfils (A3) then $\xi \mapsto \facnorm{\theta'}{\xi,1,\tc+\delta}$ is integrable outside of a compact set. 
\end{itemize}
\end{proposition}

\begin{proof}
(i): We omit the dependence on $E$ from the notation. By (A2) and (\ref{nat scale}), 
\[ \xi(z) - \xi_{\mr{crit}}(z) = \int_{z_{\mr {crit}}}^{z} (y-z_{\mr{crit}})^{\beta}(A+g(y-z_{\mr{crit}})) \, \D y = \frac{A}{\beta+1}(z-z_{\mr{crit}})^{\beta+1}(1+\Or(z-z_{\mr{crit}})).\]
Moreover, $\theta'(\xi) = \frac{\tilde p'(\xi)}{\tilde p(\xi)}=\frac{p'(x)}{p^{2}(x)}$, and thus
\[
\textstyle \theta'(\xi) = \frac{p'(z-z_{\mr{crit}})}{p^{2}(z-z_{\mr{crit}})} = \frac{\beta (z-z_{\mr{crit}})^{\beta-1}(A+\Or(z-z_{\mr{crit}}))}{A^{2}(z-z_{\mr{crit}})^{2\beta}(1+\Or(z-z_{\mr{crit}})} = \frac{\beta}{A} \frac{1}{(z-z_{\mr{crit}})^{\beta+1}}(1+\Or(z-z_{\mr{crit}})).
\]
Together, we get 
\be \label{ths form}
\theta'(\xi) = \frac{\beta}{\beta+1} \frac{1}{\xi-\xi_{\mr{crit}}} (1 + \Or(\xi-\xi_{\mr{crit}})),
\ee
which gives Assumption 2. Now if $\delta>0$ is as in Assumption (A2), and if 
$I$ is a compact interval with $\inf\{|x - \tr|: x \in I\} > \delta$, then
(\ref{ths form}) together with Darboux's theorem for power series (cf.\ \cite{BeTe2}) yields Assumption 2 with $\kappa = \sqrt{\tc^{2} + \delta^{2}} - \tc$. 
 \\
(ii): By the results of \cite{BeTe3}, it is enough to ensure that 
\[
h(\xi) = \sup \left\{ \left| \theta'(z) \right| : |z - \xi| < \tc + \delta \right\}
\]
is integrable over the real axis outside of a compact interval. By the definition of $\xi$ and the fact $\theta'(\xi) = p'(t)/p^{2}(t)$ this follows from (A3).
\end{proof}

\section{Time-Dependent Description of a Reflected Wave}

As explained in the introduction we construct solutions of the time dependent Schr\"odinger equation~(\ref{tdschINTRO}). Since we are interested in the scattering situation, we will need strong decay assumptions on the derivative of $V$, which we will give shortly. For the moment we will just assume that $V$ is bounded on the real axis, analytic in a neighbourhood of the real axis, and that $\lim_{|x| \to \infty} \partial_{x} V(x) = 0$. 

A solution to (\ref{tdschINTRO}) can now be given in the form of a wave packet: for any fixed energy $E>\max_{x\in\R}V(x)$,
let $\phi(x,E)$ be a solution to the stationary Schr\"odinger equation
\be\label{ssch}
-\eps^2\partial_x^2\phi=(E-V(x))\phi\equiv p^2(x,E)\phi,
\ee
where 
\[
p(x,E)=\sqrt{E-V(x)} >0
\]
is the classical momentum at energy $E$. By standard results on ODE we can choose 
$\phi(x,E)$ so that it is regular in both variables. 
Then the function 
\be\label{superp}
\Psi(x,t,\eps):=\int_{\Delta}Q(E,\varepsilon)e^{-itE/\varepsilon} 
\phi(x,E)\ dE,
\ee
for some regular enough energy density $Q(E,\eps)$ supported on a compact set $\Delta$, is a solution of (\ref{tdschINTRO}).
Of course, the solution to (\ref{ssch}) is not unique unless we impose boundary conditions, and it is these boundary conditions that will determine the physics that the solution $\Psi$ describes. In our case we will require that the wave packet is incoming from $x=+\infty$, as we will detail below. 

In order to connect to Section  \ref{general theorems}, we put
\[
\psi(x,E) = \left( \begin{array}{l} \phi(x,E) \\  \I \eps \phi'(x,E) \end{array} \right),
\]
and observe that $\phi$ solves (\ref{ssch}) if and only if $\psi$ solves 
\be \label{2level2}
\I \eps \partial_{x} \psi (x,E)= \pmatrix{0 & 1 \cr p^2(x,E) & 0}\psi(x,E)\equiv 
H(x,E)\psi(x,E).
\ee
Of course, this is just equation (\ref{ODE}).

Let us now fix our conditions on the potential $V$, in order to apply 
the results of Section \ref{general theorems} and to ultimately control the large $x$ behaviour of the solution to (\ref{2level2}). We select a energy window $\Delta=[E_1,E_2]$. 
The most basic assumption is\\[1mm]
{\bf (V1)} $\sup_{x} V(x) < E_{1}$, and $V$ is analytic in a neighbourhood $U$ of the real axis.\\[1mm]
This assumption corresponds to (A1) from Section \ref{general theorems}.
For the condition corresponding to (A2), we define $U_{00}(E)$ as  in (\ref{U00}). We then assume\\[1mm]
{\bf (V2)} Condition (A2) holds for $U_{00}(E)$ and for each $E \in \Delta$.\\[1mm]  
Let us see what the somewhat implicit condition (V2) means for the potential $V$. Since $p(x,E)>0$ on the real axis, the singular points of $p(z,E)$ will occur when $V(z)=E$, or when $V$ itself has a pole.  Condition (V2) states that for each $E$ there is exactly one point closest to the real axis (in the metric $d$ given in (\ref{metric})) where this happens. However, in contrast to the situation in the previous section, we now have to deal with a whole family $z_{\mr{crit}}(E)$ of such points, indexed by the energy $E$. Therefore it might happen that for different values of $E$, different exponents $\beta$ in (A2) occur, and $E \mapsto z_{\mr{crit}}(E)$ needs not be continuous. 
In order to avoid such awkward situations, we impose \\[1mm]
{\bf (V3)} The function $E \mapsto z_{\mr{crit}}(E)$ is $C^{3}$ on $\Delta$, and $\beta$ does not depend on $E$ in (V2).   
\\[1mm]  
The most natural situation is when $V(z_{\mr{crit}}) = E$ with $V'(z_{\mr{crit}}) \neq 0$. Then $z_{\mr{crit}}(E)$ depends  analytically on $E$. Another interesting situation is when $V$ has a pole; then $z_{\mr{crit}}(E)$ does not depend on $E$, but $\xi_{\mr{crit}}$ does via (\ref{xinat}).\\
Finally, we need the analog to the integrability condition (A3), uniformly in $E$; indeed, in view of some $L^{2}$ estimates to follow later on, we will need quite a bit more than (A3). We will say that $V$ fulfils (V4)$_{m}$ with $m>0$ if the following holds:  
\\[1mm]
{\bf (V4)$_{\mathbf{m}}$} Define $R_{0}(x)$ and $C_{R_{0}}(x)$ as we did just before (\ref{U00}). We write 
\[
\bk{x} = \sqrt{1+x^{2}}
\]
and assume that 
\begin{equation} \label{V4a}
h(x,\delta) : = \sup_{E \in \Delta} \sup_{z \in C_{R_{0} + \delta}(x)} \left| V'(z) (E-V(z))^{-3/2}\right| \leq \frac{C}{<x>^{3/2 + \nu}}
\end{equation}
for some $\delta, \nu > 0$, and all $x$ outside a compact interval.
Moreover, we assume that there exists $C>0$ such that for all sufficiently large $x \in \R$ we have
\[
|V'(x)| \leq C |x|^{-m-5/2}
\]
Note that (V4)$_{m}$ implies, uniformly in $E \in \Delta$,  
\be \label{scatpot}
|p(x,E) - p(\infty,E)| \sim \bra x \ket^{-3/2-m} \quad \mbox{ as }  |x| \to \infty. 
\ee
Note also that since $p(x,E)$ is bounded above and away from zero uniformly in $x$ and $E$, (A3) follows from (\ref{V4a}).  

Given the conditions on $V$, we can now apply the results of Section \ref{general theorems}. Recall the natural scale (\ref{nat scale}) 
and define $\tilde p$ through $\tilde p(\xi(x,E),E) = p(x,E)$ and $\theta'(\xi,E) =  \partial_{\xi} \tilde p(\xi,E) / \tilde p(\xi,E)$. 
By Proposition \ref{omega conditions} and Assumption (V2), $\theta'$ has exactly one pair of first order poles closest to the real line, which we denote by $\tr(E) \pm \I \tc(E)$. Now it follows from Theorem \ref{n-th} that in the $n$-th superadiabatic representation equation (\ref{2level2}) reads 
\begin{equation} \label{superadODE2} 
\I \eps \partial_{\xi}  \psi_{n}(\eps,\xi,E) = \left( \begin{array}{cc} \rho_{n}(\eps,\xi,E) & \eps^{n+1} k_{n}(\eps,\xi,E) \\[2mm] - \eps^{n+1} \overline{{k}_{n}}(\eps,\xi,E) & - \rho_{n}(\eps,\xi,E) \end{array}\right)  
\psi_{n}(\eps,\xi,E),
\end{equation}
with $\rho_{n}$ and $k_{n}$ as in Theorem \ref{n-th}. Let us write 
\[
\left( \begin{array}{l} c_{1}(\eps,x,E) \\  c_{2}(\eps,x,E) \end{array} \right) :=  \psi_{n}(\eps,\xi(x),E) = (T_{n} \tilde \psi) (\eps,\xi(x),E)
\]
for the solution of (\ref{superadODE2}) with initial condition as in Theorem \ref{solution}. 
If we choose $n_{\eps} = \tc/\eps$ so that we are in the optimal superadiabatic representation, and translate back to the actual solution  of (\ref{ssch}), we find 
\be \label{phisol}
\phi(x,E) = \left( T_{n_{\eps}}^{-1} \left( \begin{array}{l} c_{1}(\eps,x,E) \\  c_{2}(\eps,x,E) \end{array} \right)\right)^{(1)} = c_{1}(\eps,x,E) \varphi^{(1)}_{\eps,1} + c_{2}(\eps,x,E) \varphi^{(1)}_{\eps,2},
\ee
where the superscript means that we take the first component, and 
$\varphi_{\eps,j}$ are the ``optimal superadiabatic vectors'', i.e.\ the columns of $T_{n_{\eps}}$. Note that due to (\ref{UminusUnull}) and (V4) we have 
\be \label{35a}
\varphi_{\eps,2}^{(1)}(\eps,x,E) = \frac{1}{\sqrt{p(x,E)}} + \Or(\eps <x>^{-(3/2+\nu)}).
\ee
So by switching to the $n$-th superadiabatic representation we can write the solution of (\ref{ssch}) as a sum of two terms, which at the moment appear rather arbitrary. However, we will see now that, with the above initial conditions, the first one corresponds to a transmitted wave while the second one corresponds to an exponentially small reflected wave. 

To this end, note that by Theorem \ref{solution} we obtain
\[
c_{1}(\eps,x,E) = \E^{\frac{-\I}{\eps} \int_{0}^{x} p(r,E) \, \D r} (1+ \Or(\eps)),\]
and 
\be \label{c2 leading}
c_{2}(\eps,x,E) =2
\sin\left(\frac{\pi\gamma}{2}\right)
\E^{-\frac{\tc(E)}{\eps}}\E^{\frac{\I(\xi(x,E) - 2 \tr(E)}{2\eps})} {\rm
erf}\left(\frac{\xi(x,E) - \tr(E)}{\sqrt{2\eps\tc(E)}}\right)
+\Or(\eps^{\frac{1}{2}-\alpha}\E^{-\frac{\tc}{\eps}})\,,
\ee
where $\gamma = \beta/(\beta+1)$ is independent of $E$ by (V3).
In the next section we will have to investigate $x$-behaviour of the error term very closely, but for the moment we do not worry about it. 

Since the vectors $\varphi_{\eps,j}$ approach the adiabatic vectors as $|x| \to \infty$, they are, to leading order, non-oscillatory at infinity. By considering the oscillations of the $c_{j}$ and taking into account the fact that $\xi(x,E)$ approaches a linear function at $x = \pm \infty$, we find that the solution (\ref{phisol}) corresponds to 
a wave incoming from the right at $x\simeq +\infty$, described 
essentially by
\be
\frac{e^{-\I \int_0^{+\infty}(p(y,E,\eps)-p(+\infty,E))/\eps}}
{\sqrt{p(\infty,E)}}e^{-ip(+\infty,E)x/\eps},
\ee
which is scattered on the potential. The reflected wave, for large values 
of $x$ behaves as 
\be
|c_2(\eps,+\infty,E)|\frac{e^{\I \int_0^{+\infty}(p(y,E,\eps)-p(+\infty,E))/\eps}}{\sqrt{p(\infty,E)}}
e^{ip(+\infty,E)x/\eps}.
\ee
When integrated against a suitable energy density $Q(E,\eps)$ like the one we will specify below, 
this reflected wave gives rise to a freely propagating Gaussian, for large values
of $x$ and $t$, as can be shown by the methods of \cite{jm05}. 

Here, we are interested in the full history of the exponentially small reflected wave, which is only a part of the full solution that we would get by plugging (\ref{phisol}) into (\ref{superp}). From the discussion above we see that 
\be\label{superp reflect}
\chi(\eps,x,t):=\int_{\Delta}Q(E,\varepsilon)e^{-itE/\varepsilon} 
c_{2}(\eps,x,E)\ \varphi_{\eps,2}^{(1)}(\eps,x,E) dE
\ee
gives the correct asymptotic behaviour of a reflected wave, and thus we define:

\begin{definition} \label{ref wav def}
The function $\chi(x,t,\eps)$ as given in (\ref{superp reflect}) is called the reflected part of the wave in the optimal superadiabatic representation. 
\end{definition}

Let us now formulate the main result of this work. We start by fixing the assumptions on the energy density $Q$. 
We assume that $Q$ is given by  
\be
Q(E,\varepsilon)\ =\ e^{-G(E)/\varepsilon}\
e^{-\,i\,J(E)/\varepsilon}\ P(E,\varepsilon),
\ee 
where:\\[1mm]
{\bf (C1)} The real-valued function $G\ge 0$ is in
$C^3(\Delta)$ and is independent of $\eps$. Moreover, the strictly positive function 
\be \label{RM}
M(E) = G(E) + \tc(E)
\ee
has a unique non-degenerate absolute minimum at $E^{\ast} \in ]E_1, E_2[$. 
This implies that
$$
M(E) =  M(E^{\ast}) + \frac{M''(E^{\ast})}{2} (E-E^{\ast})^2 + O(E-E^{\ast})^3,\quad\mbox{ with }\quad
M''(E^{\ast})>0.
$$
{\bf (C2)} The real-valued function $J$ is in $C^3(\Delta)$.\\
{\bf (C3)} The complex-valued function $P(E,\varepsilon)$ is in
$C^1(\Delta)$ and satisfies $\lim_{\eps \to 0} P(E,\eps) = P(E,0)$ and 
\be \sup_{E\in\Delta
\atop \varepsilon \geq 0}\ \left|\,\frac{\partial^n}{\partial
E^n}\,P(\varepsilon,E)\,\right|\ \leq\ C_n,\qquad\mbox{for}\qquad
n=0,\,1.
\ee

The first of our results shows how to approximate $c_{2}$ in (\ref{superp reflect}) up to a small relative error. 
It should be viewed as an extension of Theorem \ref{solution} that includes control of $L^{2}(\D x)$ norms. Let us 
introduce the shortcut
\be  \label{a(x,E)} 
a(x,E)=\frac{\xi(x,E)-\xi_r(E)}{\sqrt{2\xi_c(E)}}=\frac{\int_{x_r(E)}^{x}p(y,E)\,dy}{\sqrt{2\xi_c(E)}}.
\ee
for the $\eps$-independent part of the argument of the error function in (\ref{c2 leading}). The point $x_r(E)$ actually marks the intersection of the Stokes line emanating from $\tc(E)$ with the real axis.
Let $E^\ast$ be as in (C1), and let us write $x^\ast_{\mr r} = x_{\mr r}(E^\ast)$. The point $x^\ast_{\mr r}$ is the birth region of the reflected wave. We define for $0<\delta<1/2$
\[
c_{\mr{eff}}(\eps,x,E) = \left\{ \begin{array}{ll}
0 & \mbox{if } x - x^\ast_{\mr r} < - \eps^{1/2-\delta}, \\
2 \sin \left( \frac{\pi \gamma}{2} \right) \e^{-\frac{\tc(E)}{\eps}} \E^{ \frac{\I}{2 \eps} \xi(x,E) - \frac{\I}{\eps} \tr(E)} \mr{erf} \left( \frac{a(x,E)}{\sqrt{\eps}} \right) & \mbox{if } | x - x^\ast_{\mr r} | \leq \eps^{1/2 - \delta},\\
2 \sin \left( \frac{\pi \gamma}{2} \right) \e^{-\frac{\tc(E)}{\eps}} \E^{ \frac{\I}{2 \eps} \xi(x,E) - \frac{\I}{\eps} \tr(E)} & 
\mbox{if } x - x^\ast_{\mr r} >  \eps^{1/2-\delta},
\end{array} \right.
\]
and 
\[
\chi_{\mr{eff}}(\eps,x,t) := \int_{\Delta} Q(E,\varepsilon) \E^{-\I tE/\varepsilon} 
c_{\mr{eff}}(\eps,x,E)\ \varphi_{\eps,2}^{(1)}(\eps,x,E) \D E.
\]

\begin{theorem} \label{chi eff thm}
Assume that $V$ fulfills (V1)--(V4)$_{m}$ with $m>6$, and that $Q$ fulfils (C1)--(C3). Then there exists $C>0$, $\eps_{0} > 0$ and $\delta > 0$ such that for all $\eps < \eps_{0}$ and all $t \in \R$ we have
\[
\norm[L^{2}(\D x)]{\chi_{\mr{eff}}(\eps,. \, ,t) - \chi(\eps,. \, ,t)} \leq C \e^{- \frac{M(E^\ast)}{\eps}} \eps^{\frac{3}{4} + \delta}
\]
\end{theorem}

The proof of this theorem will be given in the first half of Section \ref{s4}. 
Knowing now how much $\chi_{\mr{eff}}$ differs from $\chi$, we will need to show next that $\chi_{\mr{eff}}$ itself has an $L^{2}$ norm
greater than that error term. Indeed, we are going to show more: we will give the precise asymptotic shape of the 
reflected wave for almost all times $t$ after its birth time. To do so, let us introduce some additional abbreviations:
\be \label{abbrev}
R(E) := - \tr(E) - J(E), \quad  S(E,x,t) := -\frac{1}{2} \xi(x,E) - R(E) + Et,  
\ee
For an explicit description of $\chi_{\mr{eff}}$, we have to consider three different regions of the variable $x$. The first formula will be valid near the birth region of the reflected wave. We define 
\begin{eqnarray*}
\lefteqn{
\chi_{\mr{near}} (\eps,x,t) := \sqrt{\eps} P_0(x,E^{\ast})   
\e^{-\frac{M(E^\ast)+iS(E^{\ast},x,t)}{\eps}} \times} \\ 
&& \times  
\int_{\mathbb R}
\mr{erf}\left(\frac{a(x,E^*)}{\sqrt{\eps}}+a'(x,E^\ast)y\right)
\e^{-y^{2} \frac{M''(E^\ast)+iS''(E^\ast,x,t)}{2}} \e^{-\frac{\I}{\sqrt{\eps}} S'(E^{\ast},x,t) y}
\,dy. 
\end{eqnarray*}
with 
\[ P_{0}(x,E^{\ast}) = \frac{2 P(0,E^{\ast}) \sin \left( \frac{\pi \gamma}{2} \right)}{\sqrt{p(x,E^\ast)}}
\]
Here and henceforth, primes denote derivatives with respect to $E$.

In the region of moderately large $x$, but away from the transition point, the reflected wave will be given by 
\be \label{chimod}
\chi_{\mr{mod}}(\eps,x,t)=  \frac{P_{0}(x,E^\ast) \sqrt{2 \pi \eps}}{\sqrt{M''(E^{\ast}) + \I S''(E^{\ast},x ,t)}}
\e^{-\frac{M(E^{\ast}) + \I S(E^{\ast},x,t)}{\eps}} \e^{-\frac{ S'(E^{\ast},x,t)^{2}}{2 \eps (M''(E^{\ast}) + \I S''(E^{\ast},x,t))}}.
\ee

For the asymptotic regime, i.e. for large $x$, the reflected wave is characterized by its asymptotic momentum 
\[
k(E) := \sqrt{E - V(\infty)}.
\]
rather than its energy. 
We put $k^{\ast} = k(E^{\ast})$, write $\tilde f(k)=f(E(k))$ for any function $f$, and define
\[
\om(E) = \int_{0}^{\infty} (p(y,E)-p(\infty,E)) \, \D y.
\]
The reflected wave in the asymptotic region is then described by \\
\begin{eqnarray} \label{xfar}
\chi_{\mr{far}}(\eps,x,t)&=& \frac{2 \sin(\gamma \pi /2) P(0,E^\ast)  2\sqrt{2\pi\eps k^{\ast}}} {\sqrt{\tilde\M''(k^\ast)+i(2t-(\tilde R+\tilde \omega)''(k^\ast))}} \e^{-\frac{\M(E^\ast)}{\eps}} \e^{\frac{\I}{\eps} (R(E^\ast)+\omega(E^\ast) +xk^\ast-E^\ast t)}\times \nonumber
\\ 
&&\times \exp\left(-\frac{(x-2k^\ast t+(\tilde R+\tilde \omega)'(k^\ast))^2}{2\eps
(\tilde\M''(k^\ast)+i(2t-(\tilde R+\tilde \omega)''(k^\ast)))}\right).
\end{eqnarray}

Let us now give a precise statement about the shape of the reflected wave. 

\begin{theorem} \label{reflected wave shape}
Assume that $V$ fulfills (V1)--(V4)$_{m}$ with $m>10$, and that $Q$ fulfils (C1)--(C3). 
Let 
$0<\delta<1/2$, $C_0 > 0$, and  $C_1>0$. For $0 < \beta < 1/8$ define 
\[
\chi_{\mr{expl}}(\eps,x,t) = \left\{ \begin{array}{ll}0 &  \mbox{if } x - x_{\mr r}(E^\ast) \leq - C_0 \eps^{1/2 - \delta} \\
\chi_{\mr{near}}(\eps,x,t) & \mbox{if } |x - x_{\mr r}(E^\ast)| < C_0 \eps^{1/2 - \delta} \\
\chi_{\mr{mod}}(\eps,x,t)  & \mbox{if } C_0 \eps^{1/2 - \delta} \leq x - x_{\mr r}(E^\ast) < C_1 \eps^{-\beta}\\
\chi_{\mr{far}}(\eps,x,t) &\mbox{if } C_1 \eps^{-\beta} \leq x - x_{\mr r}(E^\ast) .
\end{array} \right.
\]
Then there exist $\beta_0 < 1/8$ and $\eps_0 > 0$ and $D>0$ such that for all $\beta$ with $\beta_0 < \beta < 1/8$ and all $0 < \eps < \eps_0$ we have 
\[
\norm[L^2(\D x)]{\chi_{\mr{expl}}(\eps,. \, ,t) - \chi_{\mr{eff}}(\eps,. \, ,t)} \leq D \eps^{3/4 + \delta_0}e^{-M(E^*)/\epsilon}
\]
for some $\delta_0 > 0$.
\end{theorem}

It is obvious that $\chi_{\mr{far}}$ is a Gaussian wave packet. Our next result shows that this is also true, to leading order, for all other times just after the transition occurred. The center of the wave packet is the classical trajectory $q_t$ of the reflected wave, where $q_t$ is defined to be the unique solution of the equation $S'(E^\ast,q_t,t) = 0$. 
In particular, $q_t$ satisfies
 $
 \frac{1}{2}\int_0^{q_t}\frac{1}{p(y,E^*)}dy+R'(E^*)=t,
 $
 which shows that $q_t$ behaves like $t-R'(E^*)$.

\begin{theorem} \label{chimod gauss}
Define
\[
\chi_{\mr{gauss}} = \frac{P_0(E^\ast,q_t) \sqrt{2 \pi \eps}}{\sqrt{M''(E^\ast) + \I S''(E^\ast,q_t,t})} \e^{-\frac{M(E^\ast) + \I S(E^\ast,x,t)}{\eps}} \e^{\frac{-(x-q_t)^2}{8 \eps p(E^\ast,q_t)^2 (M''(E^\ast) + \I S''(E^\ast,q_t,t))}}.
\]
Assume $t$ is such that $\eps^{1/2 - \delta} < |q_t - x^\ast_{\mr r}| < \eps^{- \beta}$, with $\beta< 1/8$ as in Theorem \ref{reflected wave shape}. Then there exist $C>0$, $\eps_0>0$ and $\delta_0 > 0$ such that 
\[
\norm[L^2]{\chi_{\mr{gauss}}(\eps,. \, , t) - \chi_{\mr{eff}}(\eps,. \, , t) } \leq C \e^{-\frac{M(E^\ast)}{\eps}} \eps^{3/4 + \delta_0}
\]
for all $\eps < \eps_0$.
\end{theorem}

Like the two previous theorems, we will prove Theorem \ref{chimod gauss} in Section \ref{s4}.
A direct calculation of Gaussian integrals now yields

\begin{corollary}
With the definitions above, we have
\[
\norm[L^2]{\chi_{\mr{gauss}}(\eps,. \, , t)} = \norm[L^2]{\chi_{\mr{far}}(\eps,. \, , t)}
= \frac{4 |P(0,E^\ast)\sin (\pi \gamma / 2)|\pi^{3/4}}{(M''(E^\ast))^{1/4}} 
\e^{-\frac{M(E^\ast)}{\eps}} \eps^{3/4}.
\]
Thus, for all times such that $\eps^{1/2 - \delta} < q_t - x^\ast_{\mr r} $,  Theorem \ref{reflected wave shape} yields an approximation $\chi_{\mr{expl}}$ of the reflected
wave $\chi$ whose norm is larger than the error term.

\end{corollary}
\noindent
{\bf Remark:} We make use of the  identity $\tilde M''(k^\ast)=M''(E^\ast)4{k^\ast}^2$,
which follows from the fact that $E\mapsto M(E)$ is critical at $E^\ast$.

\section{Rigorous analysis of the reflected wave} \label{s4}

In this section we give the proof of Theorems \ref{chi eff thm} -- \ref{chimod gauss}. 
It is checked directly that 
\begin{equation} \label{c2 full}
c_{2}(\eps,x,E) = \frac{\I }{\eps}\,\E^{\frac{\I}{\eps}
\int_{-\infty}^{\xi(x)} \D y \,\rho_{n}(y)} \int_{-\infty}^{\xi(x)}
\E^{-\frac{\I}{\eps} \int_{-\infty}^y
\rho_{n}(r) \, \D r} \, ( \eps^{n + 1} \overline{k}_{n_{\eps}}(\eps, y, E)) \, \psi_{+,n}(\eps,y,E) \, \D y,
\end{equation}
solves the second line of (\ref{superadODE2}) with initial condition as in Theorem \ref{solution}. In what follows, we will always understand $n$ to be the optimal one given by (\ref{ndef}). 

The additional difficulty of the present problem when compared to \cite{BeTe1, BeTe2} is that we now need to extract the leading term of (\ref{c2 full}) with errors that are not only small pointwise in $x$, but even integrable over all $x \in \R$. In particular, neither the leading term of $c_2(\eps,x,E)$ nor all of its error terms are small when $x \to +\infty$, and it is only after integration over $E \in \Delta$ that these terms become integrable, due to their oscillatory nature. This determines our strategy: 
we will always decompose the error term into a part that decays as $x \to \infty$ and one that oscillates there in a regular way, both of which will give rise to a small contribution to the $L^{2}$-norm. More precisely, we introduce 
\begin{definition} Let us write $\bk{x} = \sqrt{1 + x^2}$.
\begin{itemize}
\item[(i)]
We say that a function $r = r(\eps,x,E)$ is in $\ec_{1}(\beta,\nu)$ for $\beta \in \R$ and $\nu>0$ if there exists $C$ independent of $E$ such that
\[
|r(\eps, x,E)| \leq C \eps^{\beta} \bk{x}^{-1/2-\nu}.
\]
\item[(ii)]
We say that the function $r$ is in $\ec_{2}(\beta)$ for $\beta \in \R$ if there exists $C<\infty$ independent of $E$  and  $r_{0}(\eps,E)$ independent of $x$ and with $|r_{0}(\eps,E)| \leq C\eps^{\beta}$ and 
\[
r(\eps,x,E) = \left\{ \begin{array}{ll} \E^{\frac{\I}{\eps} x p(\infty,E)} r_{0}(\eps,E) & \mbox{if } x>0\\
										0 & \mbox{if } x<0. \end{array} \right.
\]
Here, we put $p(\infty,E) = \lim_{x \to \infty} p(x,E) = \lim_{x \to \infty} \sqrt{E - V(x)}$.
\end{itemize}
\end{definition}
We immediately show how to use these two different functions in order to estimate $L^{2}$ errors.

\begin{lemma} \label{L00} 
Let $\Delta$ be a compact interval and $f: \Delta \to \R$ be a bounded function. For $r = r(\eps,x,E)$ define 
\[ I(x) = \int_{\Delta} f(E) r(\eps,x,E) \, \D E. \]
If $r \in \ec_{1}(\beta,\nu)$ for some $\nu>0$ then there exists $C>0$ such that 
\[
\norm[L^{2}(\D x)]{I} \leq C \eps^{\beta} \norm[L^{1}(\D E)]{f} \]
If $r \in \ec_{2}(\beta)$ then there exists $C>0$ such that
\[ \norm[L^{2}(\D x)]{I} \leq C \eps^{\beta+1/2} \norm[L^{2}(\D E)]{f}.\]
\end{lemma}
\begin{proof}
The first assertion is just an application of Fubini's theorem together with the fact that a function from $\ec(\beta,\nu)$ is square integrable in $x$.
The second assertion follows from the fact that with 
\[ J(x) = \int_{\Delta} \E^{\frac{\I}{\eps} p(\infty,E)x} f(E)  \, \D E \]
we have 
\be \label{explicit integral} 
\int |J(x)|^{2} \, \D x = 4 \pi \eps \int_{\Delta} p(\infty,E) |f(E)|^{2} \, \D E.
\ee
To see this, put $k = p(\infty,E) = \sqrt{E - V(\infty)}$, so that $E = k^{2} + V(\infty)$ and $\D E = 2 k \, \D k$. Putting $\tilde \Delta = p(\infty, \Delta)$ we have 
\[
J(x) = \int_{\R} \E^{\frac{\I}{\eps} k x} 1_{\tilde \Delta}(k) f(k^{2} + V(\infty)) 2 k \, \D k 
= \sqrt
{2 \pi \eps} \left( \mathcal F_{\eps}^{-1} \rho \right)(x),\]
where $\rho(k) = 2 k 1_{\tilde \Delta}(k) f(k^{2} + V(\infty))$ and $\mathcal F_{\eps}$ is 
the rescaled Fourier transform
\be
({\mathcal F}_\eps g)(k)=\frac{1}{\sqrt{\eps2\pi}}\int_{\mathbb R}g(x)e^{-ikx/\eps}\,dx.
\ee 
By the Plancherel formula we find 
\[
\norm[L^{2}(\D x)]{\mathcal F_{\eps}^{-1} \rho}^{2} = \norm[L^{2}(\D k)]{\rho}^{2} = 4  \int_{\tilde \Delta} k^{2} |f(k^{2} + V(\infty))|^{2} \, \D k = 2 \int_{\Delta} \sqrt{E-V(\infty)} |f(E)|^{2} \, \D E.
\]
Taking into account the prefactor $\sqrt{2 \pi \eps}$ we obtain (\ref{explicit integral}).
\end{proof}

Let us now start to see how to decompose quantities of interest into parts from $\ec_{1}$ and $\ec_{2}$. 

\begin{lemma} \label{L01}
Assume that $V$ fulfils (V4)$_{m}$ for some $m>0$. Then for each $\beta>0$ and each $0 < \nu < m$ we can find $r_{1} \in \ec_{1}(\beta(m-\nu)-1,\nu)$ and $r_{2} \in \ec_{2}(0)$ such that 
\[
1_{\{x > \eps^{-\beta}\}} \E^{\frac{\I}{2 \eps} \xi(x,E)} = r_{1} + r_{2}.
\]
In fact, we can choose
\[
r_{2}(\eps,x,E) = 1_{\{x > \eps^{-\beta}\}} \E^{\frac{\I}{\eps} (\om(E) + p(\infty,E) x)},
\]
with $\om(E) = \int_{0}^{\infty} (p(y,E) - p(\infty,E)) \, \D y$.
\end{lemma}
\begin{proof}
We decompose 
\[
\xi(x,E) = 2 \int_{0}^{x} p(y,E) \, \D y = 2 \int_{0}^{\infty} (p(y,E) -  p(\infty,E)) \, \D y + 2 K(x,E) + 2 x p(\infty,E),
\]
where 
\[
K(x,E) = \int_{x}^{\infty} (p(y,E)-p(\infty,E)) \, \D y.
\]
By (V4)$_{m}$ we have $|K(x,E)| \leq C \bk{x}^{-1/2-m}$ for $x>0$, and from the inequality $|\e^{\I x}-1| \leq |x|$ valid for all $x \in \R$ we get 
\[
\left| 1 - \E^{\frac{\I}{2 \eps} K(x,E)} \right| \leq \frac{C}{\eps} \bk{x}^{-1/2-m} = C \frac{\bk{x}^{-m+\nu}}{\eps} \bk{x}^{-1/2-\nu} = (\ast).
\]
Since $\bk{x} > \eps^{-\beta}$ by assumption, we have $\bk{x}^{-m+\nu} < \eps^{\beta(m-\nu)}$ and thus\\ 
$(\ast) \leq C \eps^{\beta(m-\nu) -1} \bk{x}^{-1/2-\nu}$, which proves the claim.
\end{proof}

\begin{corollary} \label{C01}
Assume that $V$ fulfils (V4)$_{m}$ for some $m>0$. Then for each fixed $D \in \R$, and for each $0 < \nu < m$, there exist 
\[
r_{1} \in \ec_{1} \left( - \frac{2 + \nu}{2+m},\nu \right), \quad r_{2} \in \ec_{2}(0)
\]
with 
\[
1_{\{x>D\}} \E^{\frac{\I}{2 \eps} \xi(x,E)} = r_{1} + r_{2}.
\]
\end{corollary}
\begin{proof} 
For $\beta>0$ and $x$ with $D < x < \eps^{-\beta}$ we have 
\[
1_{\{D < x < \eps^{-\beta}\}} \leq 1_{\{D < x < \eps^{-\beta}\}} \bk{x}^{1/2+\nu} \bk{x}^{-1/2-\nu} < C \eps^{-\beta(1/2 + \nu)} \bk{x}^{-1/2-\nu},
\]
which is in $\ec_{1}(-\beta(1/2 + \nu),\nu)$. From Lemma \ref{L01} we know that the part with $x > \eps^{-\beta}$ is in $\ec_{1}(\beta(m - \nu) -1,\nu) + \ec_{2}(0)$. Optimising over $\beta$ gives 
the condition $-\beta(1/2 + \nu) = \beta(m-\nu) - 1$, whence $\beta = 1/(1/2+m)$ and $-\beta(1/2+\nu) = - (2+\nu)/(2+m)$.
\end{proof}

\begin{corollary} \label{C02}
Assume that $V$ fulfils (V4)$_{m}$ for some $m>0$, fix $D \in \R$ and $0 < \nu < m$. Assume that $f(\eps,x,E)$ is such that 
\begin{equation} \label{C02 condition}
\left| \int_{-\infty}^{\infty} f(\eps,x,E) \, \D x \right| < C \eps^{\beta}, \quad \mbox{and} \quad 
1_{x>D} \left| \int_{\xi(x)}^{\infty} f(\eps,x,E) \, \D x \right| \leq C \eps^{\beta} \bk{x}^{-\frac{1}{2}-\nu}
\end{equation}
for some $C>0$ and some $\beta \in \R$. 
Then there exist  
\[
r_{1} \in \ec_{1} \left( \beta - \frac{2 + \nu}{2+m},\nu \right), \quad r_{2} \in \ec_{2}(\beta)
\]
such that 
\[
1_{\{x>D\}} \E^{\frac{\I}{2 \eps} \xi(x,E)} \int_{-\infty}^{\xi(x,E)} f(\eps,s,E) \, \D s = r_{1} + r_{2}.
\]
\end{corollary}
\begin{proof} 
Simply write $\int_{-\infty}^{\xi(x)} f \, \D s = \int_{-\infty}^{\infty} f \, \D s - \int_{\xi(x)}^{\infty} f \, \D s$, and use Corollary \ref{C01}.
\end{proof}

We will now apply the above tools to the study of (\ref{c2 full}). Note that 
\be \label{xi asymp}
\lim_{x \to \pm \infty } \frac{\xi(x)}{x} = \lim_{x \to \pm \infty} \sqrt{E-V(x)}
\ee
which is finite by assumption, so we can replace $x$ with $\xi$ for studying the decay rate.

\begin{proposition} \label{P01}
Assume that $V$ fulfils (V1)--(V4)$_{m}$ with $m>0$. For $D$ sufficiently negative define 
\[
c_{2,D}(\eps,x,E) =  1_{\{x>D\}} \frac{\I}{\eps} \e^{\frac{\I}{2\eps} \xi(x)} \int_{-\infty}^{\xi(x)} \e^{-\frac{\I}{\eps} y}  ( \eps^{n + 1} \overline{k}_{n_{\eps}}(\eps, y, E)) \, \D y.
\]
For each $\nu_{0} > 0$ there exists $0 < \nu < \nu_{0}$, and 
\[
r_{1} \in \ec_{1} \left( \frac{1}{2} - \frac{2 + \nu}{2+m},\nu \right), \quad r_{2} \in \ec_{2}\left( \frac{1}{2} \right)
\]
such that 
\[ |c_{2,D}(\eps,x,E) - c_{2}(\eps,x,E)| = (r_{1} + r_{2}) \e^{-\frac{\tc(E)}{\eps}}.\]
\end{proposition}
\begin{proof}
First note that by (V4)$_m$ and (\ref{rhoc2}), for sufficiently negative $y$ we have 
\[ | \eps^{n+1} k_{n}(\eps,y,E) | \leq \E^{-(\xi_c+\xi_c\ln(\kappa /\xi_c))/\epsilon}
\bk{y}^{-(3/2 + \nu)} \leq \eps^{\infty} \E^{-\frac{\tc(E)}{\eps}} \bk{y}^{-(3/2 + \nu)}.\]
Using this in (\ref{c2 full}) gives 
\[
| 1_{\{x<D\}} c_{2}(\eps,x,E)| \leq C \eps^{\infty} \E^{-\frac{\tc(E)}{\eps}} \bk{x}^{-(1/2 + \nu)} = \E^{-\frac{\tc(E)}{\eps}} r
\]
with $r \in \ec_{1}(\infty,\nu)$. Now let us define $K(\eps,x,E) = \int_{x}^{\infty} (\rho_{n}(y) - 1/2) \, \D y$ which is finite and $\Or(\eps^{2})$ for all $x$ by (\ref{rhoc1}). We write 
\begin{equation} \label{P2.0}
\int_{0}^{\xi(x)} \rho_{n}(y) \, \D y = \frac{1}{2} \xi(x,E) - K(\eps, \xi(x),E) + K(\eps,0,E).
\end{equation}
Now for $x>D$ we have 
\begin{equation} \label{P2.1}
|K(\eps,\xi(x),E)| \leq \eps^{2} \int_{\xi(x)}^{\infty} \bk{y}^{-3/2-\nu} \, \D y \leq C \eps^{2} \bk{x}^{-1/2-\nu},
\end{equation}
by (\ref{rhoc1}). Thus 
\begin{equation} \label{50b}
1_{\{x>D\}} \E^{- \frac{\I}{\eps} K(\eps,\xi(x),E) } = 1_{\{x>D\}} (1 + \eps r(x,\eps,E))
\end{equation}
with $|r(x,\eps,E)| \leq \bk{x}^{-\frac{1}{2}-\nu}$. Using (\ref{rhoc2}), (\ref{gdef}) and (V4), we find 
\be \label{unibound}
|\eps^{n+1} k_{n}(\eps,x,E)| \leq C \sqrt{\eps} \E^{-\frac{\tc}{\eps}} \bk{x}^{-3/2-\nu},
\ee
and (\ref{50b}) together with (\ref{unibound}) give
\begin{equation}\label{P2.2}
1_{\{x>D\}} c_{2} = 1_{\{x>D\}} \frac{\I}{\eps} \E^{\frac{\I}{2 \eps} \xi(x)} \int_{-\infty}^{\xi(x)} \E^{\frac{\I}{\eps} \left( \int_{-\infty}^{y} \rho_{n}(r) \, \D r + K(\eps,0,E) \right)} \eps^{n+1} k_{n} \psi_{+,n} \, \D y  + \e^{-\frac{\tc}{\eps}} r_{1}
\end{equation}
with $r_{1} \in \ec_{1}(\frac{1}{2},\nu)$. Inside the integral, (\ref{upper}) gives
\[
\e^{-\frac{\I}{\eps} \int_{-\infty}^{y} \rho_{n}(r) \, \D r} \psi_{+,n}(y,r) = \e^{-\frac{2 \I}{\eps} \int_{-\infty}^{y} \rho_{n}(r) \, \D r} (1 + \Or(\eps^{\infty})),
\]
and using (\ref{rhoc1}) we see
\[
\e^{-\frac{2 \I}{\eps} \int_{-\infty}^{y} \rho_{n}(r) \, \D r + \frac{\I}{\eps} K(\eps,0,E)} = \e^{-\frac{\I}{\eps} y} + \Or(\eps).
\]
Combining this with (\ref{unibound}) and (\ref{P2.0}) we obtain 
\begin{eqnarray*}
&&\E^{\frac{\I}{\eps} \left( \int_{-\infty}^{y} \rho_{n}(r) \, \D r + K(\eps,0,E) \right)} \eps^{n+1} \overline{k}_{n}(\eps,y,E) \psi_{+,n}(\eps,y,E) = \\ 
&& = \e^{-\frac{\I}{\eps} y} \eps^{n+1} \overline{k}_{n}(\eps,y,E) + \eps^{3/2} \e^{-\frac{\tc}{\eps}} f(\eps,y,E),
\end{eqnarray*}
where $|f(\eps,y,E)| \leq \bk{x}^{-3/2-\delta}$. When inserting this into (\ref{P2.2}), the first term above gives $c_{2,D}$, while the second one is estimated using Corollary \ref{C02} to give the desired result.
\end{proof}

In order to proceed further, we now need the following refinement of Theorem \ref{MainTh} 
which gives us additional control over the $x$-dependence of the error terms to $k_{n_{\eps}}$. 
\begin{proposition} \label{remainder x}
For each $\alpha>0$ there exists 
$r \in \ec_{1}(3/2-\alpha, 1+\nu)$ such that 
\[ \eps^{n_{\eps}+1} k_{n_{\eps}}(\eps,\xi,E) = g(\eps,\xi,E) + \E^{-\frac{\tc}{\eps}} r(\eps, \xi, E),\]
where $g$ is given through (\ref{gdef}).
\end{proposition}

\begin{proof} 
We will use Assumption (V4).  We use (\ref{Hod}) in (finite) neighbourhood of $\tr$ and (\ref{rhoc2}) outside this neighbourhood; As for (\ref{rhoc2}), we know (\cite{BeTe3}, Proposition 2)
\[
\facnorm{\theta'}{\{\xi\},1,\tc + \delta} \leq C \sup_{|z- \xi| < \tc + \delta} |\theta'(z)|,
\]
and $\theta'(\xi) = \tilde p'(\xi) / \tilde p(\xi)$; together with (\ref{xi asymp}) and (V4)$_{\mr m}$ the claim follows.
Note that we may take $\alpha$ arbitrarily small in Assumption 2, due to Proposition \ref{omega conditions}. 
\end{proof}

We will also need the following useful little lemma:

\begin{lemma} \label{gaussfact}
For each $n \in \N$ and each $\eps > 0$ we have
\[
x^n \e^{-\frac{x^2}{2\eps}} \in \ec_1(n/2,\nu) \quad \mbox{for all } \nu > 0,
\]
and, for $\delta>0$,
\[
1_{\{x>\eps^{-\delta}\}} \e^{-\frac{x^2}{2 \eps}} \in \ec_1(m,\nu) \quad \mbox{for all } m>0, \nu>0.
\]
\end{lemma}
\begin{proof}
Taking the derivative of $g(x) = x^n \e^{-\frac{x^2}{4 \eps}}$ and equalling it to $0$, we see that $g$ is maximal at $x = \sqrt{2 \eps n}$. The value at the maximum is given by $\e^{-n/2} (2 \eps n)^{n/2} = \Or(\eps^{n/2})$. On the other hand, $f(x) = \e^{-\frac{x^2}{4 \eps}} \leq \e^{-\frac{x^2}{4}} \leq C \bk{x}^{-\nu}$ for all $\nu > 0$. Since $x^n \e^{-\frac{x^2}{2}} = f(x) g(x)$ and $\e^{-\frac{x^2}{2 \eps}} = f^2(x)$, both statements are now obvious.
\end{proof}

\begin{proposition} \label{P02}
Assume $V$ fulfils (V1)--(V4)$_{m}$. Define 
\[
c_{2,L}(\eps,x,E) = 1_{\{x>D\}} 2 \sin \left( \frac{\pi \gamma}{2} \right) \e^{-\frac{\tc}{\eps}} \E^{ \frac{\I}{2 \eps} \xi(x) - \frac{\I}{\eps} \tr(E)} \mr{erf} \left( \frac{\xi(x,E) - \tr(x,E)}{\sqrt{2 \eps \tc(E)}} \right).
\]
Then for arbitrary $\alpha>0$, for $D$ negative and $|D|$ large enough, there exist 
\[
r_{1} \in \ec_{1}\left( \frac{1}{2} - \frac{2 + \nu}{2 + m} - \alpha, \nu \right), \quad r_{2} \in \ec_{2} \left( \frac{1}{2} - \alpha \right)
\]
such that 
\[
c_{2}(\eps,x,E) = c_{2,L}(\eps,x,E) + \E^{-\frac{\tc}{\eps}}(r_{1}+r_{2}).
\]
\end{proposition}
\begin{proof}
By Proposition \ref{P01} it is enough to consider $c_{2,D}-c_{2,L}$. Using Proposition \ref{remainder x} and, as before, (\ref{unibound}) and Corollary \ref{C02}, we arrive at 
\be \label{P4.0}
c_{2,D} =  1_{\{x>D\}} \frac{\I}{\eps} \E^{ \frac{\I}{2 \eps} \xi(x)} \int_{-\infty}^{\xi(x)} \E^{-\frac{\I}{\eps} y} \overline{g}(\eps,y,E) \, \D y + \e^{-\frac{\tc}{\eps}} (r_{1}+r_{2})
\ee
and $r_{1}$, $r_{2}$ as in the statement. 
Now recall that 
\be \label{P4.1}
\overline{g}(\eps,y,E) = - 2\I\,\sqrt{{\textstyle\frac{2\eps}{\pi
\tc}}}\,\sin\left({\textstyle\frac{\pi \gamma}{2}}\right)\,
\E^{-\frac{\tc}{\eps}}\,\E^{-\frac{(y-\tr)^2}{2\eps \tc}} \,\cos\left( w(\eps,y,E) \right),
\ee
with
\[
w(\eps,y,E) = \frac{y-\tr}{\eps}
-\frac{(y-\tr)^3}{3\eps\tc^2} + \frac{\sigma_\eps (y-\tr)}{\tc}.
\]
We write 
\be \label{P4.2}
\cos\left( w(\eps,y,E) \right) = 
\frac{1}{2} \left( \e^{\I w(\eps,y,E) } + 
\e^{-\I w(\eps,y,E) } \right).
\ee
For the second term above, the exponential $\e^{\frac{- \I}{\eps} y}$ present in $ \e^{-\I w(\eps,y,E) }$ combines with the one present in (\ref{P4.0}); 
we use integration by parts in (\ref{P4.0}), with the function $\e^{\frac{-2 \I}{\eps} y}$ being integrated. The result is 
\begin{eqnarray} \label{P4.3}
\lefteqn{\int_{-\infty}^{\xi} \E^{-\frac{\I}{\eps} y} \e^{-\I w(\eps,y,E) }\E^{-\frac{(y-\tr)^2}{2\eps \tc}} \, \D y = \frac{\eps}{-2\I} \E^{-\frac{\I}{\eps} \xi} \e^{-\I w(\eps,\xi,E) }\E^{-\frac{(\xi-\tr)^2}{2\eps \tc}} -} \\
&& - 
\frac{1}{-2 \I} \int_{-\infty}^{\xi} \E^{-\frac{\I}{\eps} y} \e^{-\I w(\eps,y,E) }\E^{-\frac{(y-\tr)^2}{2\eps \tc}} \left(\frac{2(y-\tr)}{\tc} - \frac{\I(y-\tr)^{2}}{\tc^{2}} + \I \eps \frac{\sigma_{\eps}}{\tc} \right) \, \D y.
\nonumber
\end{eqnarray}
The boundary term is clearly in $\ec_{1}(1,\nu)$, and the factors left out in (\ref{P4.3}) are bounded by $C \eps^{-1/2} \e^{-\frac{\tc}{\eps}}$, so together this gives a contribution in 
$\e^{-\frac{\tc}{\eps}} \ec_{1}(1/2,\nu)$. The contributions of the integrand in the bulk term stemming from the last two terms in the bracket are in ${\cal E}_1(1,\nu)$ by Lemma \ref{gaussfact}. For the first term, we have to integrate by parts once again; the resulting boundary term is again clearly in $\ec_1(1,\nu)$, while the bulk term is very similar to the one in (\ref{P4.3}) with the crucial difference that the bracket there is now multiplied by $2 (y-\tr) /\tc$ and an additional term $2 \eps/\tc$ is present. Applying Lemma \ref{gaussfact} then shows that this integrand is in ${\cal E}_1(1,\nu)$ again. Now we can apply Corollary \ref{C02} and find that 
the term involving $ \e^{\I w(\eps,y,E) }$ is bounded by $\e^{-\frac{\tc}{\eps}}(r_{1}+r_{2})$ with 
$r_{1} \in \ec_{1}(1/2,\nu)$ and $r_{2} \in {\cal E}_2(1/2)$.\\
We now turn to the term involving $ \e^{+\I w(\eps,y,E) }$. Here the oscillations $\e^{\frac{\pm \I}{\eps} y}$ cancel in the leading term of (\ref{P4.0}), which is therefore given by 
\[
c_{2,D}^{-} =  1_{\{x>D\}} 
\sqrt{{\textstyle\frac{2}{\eps \pi \tc}}}
\sin\left({\textstyle\frac{\pi \gamma}{2}}\right)\,
\E^{-\frac{\tc}{\eps}} \e^{\frac{\I}{2 \eps} \xi(x) - \frac{\I}{\eps} \tr}
\int_{-\infty}^{\xi(x)} \E^{-\frac{(y-\tr)^{2}}{2 \eps \tc}} \e^{-\I \frac{(y-\tr)^{3}}{3 \eps \tc^{2}} + \frac{\I \sigma_{\eps}(y-\tr)}{\tc}} \, \D y.
\]
We now use explicit integration 
\[
\int_{-\infty}^{\xi(x)} \e^{-\frac{(y - \tr)^2}{2 \eps \tc}} \, \D y = \sqrt{2 \pi \eps \tc} 
\mr{erf} \left( \frac{y - \tr}{\sqrt{2 \eps \tc}} \right)
\]
and the fact that $|\e^{\I x} - 1| \leq |x|$ for all real $x$ to find 
\be \label{c2-D}
c_{2,D}^{-} = c_{2,L} +  1_{\{x>D\}} \frac{1}{\sqrt{\eps}} \e^{-\frac{\tc}{\eps}} 
\e^{\I \frac{\xi(x)}{\eps}} \int_{-\infty}^{\xi(x)} \e^{-\frac{(y-\tr)^{2}}{2 \eps \tc}} f(y) \, \D y,
\ee
with 
\[
|f(y)| \leq C \left( \frac{|y-\tr|^{3}}{\eps \tc} + |y - \tr| \right).
\]
Using integration by parts, we find 
\[
\int_{-\infty}^{R} \e^{-\frac{y^2}{2 \eps \tc}} \frac{|y|^3}{\eps \tc} \, \D y = 
2\epsilon\xi_c(2-e^{-R^2/(2\epsilon\xi_c)})-e^{-R^2/(2\epsilon 
\xi_c)}R^2
\]
Now an easy calculation using Lemma \ref{gaussfact} shows that Corollary \ref{C02} applies to the second term of (\ref{c2-D}), showing that it is bounded by $\e^{-\frac{\tc}{\eps}}(r_{1}+r_{2})$ with 
$r_{1} \in \ec_{1}(1/2,\nu)$ and $r_{2} \in {\cal E}_2(1/2)$.
\end{proof}

After all these preparations, we now work towards the proof of Theorem \ref{chi eff thm}.
Recall that $\chi$ was defined in (\ref{superp reflect}), and in a similar way define
\be \label{chiL}
\chi_L(\eps,x,t):=\int_{\Delta}Q(E,\varepsilon)e^{-itE/\varepsilon} 
c_{2,L}(\eps,x,E)\ \varphi_{\eps,2}^{(1)}(\eps,x,E) dE.
\ee

\begin{proposition}
Assume that $V$ fulfills (V1)--(V4)$_m$, and that $Q$ fulfils (C1)--(C3). 
Then for all $\alpha > 0$ there exist $\eps_0>0$ and $C>0$ such that for all $\eps < \eps_0$ we have 
\[
\norm[L^2]{\chi(\eps,. \, , t) - \chi_L(\eps,. \, , t) } \leq  C  \E^{-\frac{M(E^{\ast})}{\eps}} \eps^{1 - \frac{2+\nu}{2+m} - \alpha}.
\]
In particular, for $m>6$ we have $\norm[L^2]{\chi(\eps,. \, , t) - \chi_L(\eps,. \, , t) } \leq  C  \E^{-\frac{M(E^{\ast})}{\eps}} \eps^{3/4 + \delta}$ for some $\delta > 0$. 
\end{proposition}
\begin{proof}
By Proposition \ref{P02} and Lemma \ref{L00}, we have 
\[
\norm[L^2]{\chi(\eps,. \, , t) - \chi_L(\eps,. \, , t) } \leq C \eps^{\frac{1}{2} - \frac{2 + \nu}{2 + m} - \alpha} \sup_{x,E} |\varphi_{\eps,2}^{(1)}| ( \norm[L^1(\D E)]{Q_0} + \eps^{1/2} \norm[L^2(\D E)]{Q_0} ),
\]
where $Q_0(E) = Q(E) \e^{-\frac{\tc(E)}{\eps}}$.
(\ref{35a}) ensures that $\sup_{x,E} |\varphi_{\eps,2}^{(1)} (\eps,x,E)|$ is bounded uniformly for all $\eps < \eps_0$. By virtue of (C1), for arbitrarily small $\delta > 0$ we have 
$M(E)/\eps \geq M(E^\ast)/\eps + \eps^{-2 \delta}$ when $|E-E^\ast| > \eps^{1/2 - \delta}$.
Thus writing $\Delta_\eps = \{ E \in \Delta: |E-E^\ast| < \eps^{1/2 - \delta}\}$ we find that for each $m \in \N$ there exists $C>0$ with 
\be \label{q sharp}
|Q_0(E) 1_{\{ E \notin \Delta_\eps \}}| \leq C \e^{-\frac{M(E^\ast)}{\eps}} \eps^m, 
\ee
and, by (C1) again 
$$
|Q_0(E) 1_{\{ E \in \Delta_\eps \}}|\leq C\E^{-\frac{M(E^\ast)}{\eps}}
\E^{-\frac{M''(E^\ast)(E-E^\ast)^2}{2\eps}(1+(\eps^{1/2-\delta}))}.
$$
Consequently, by scaling
\bea \label{L1 norm}
\norm[L^1]{Q_0} &\leq& C \e^{-M(E^\ast)/\eps} \int_{\mathbb R}\E^{-\frac{M''(E^\ast)(E-E^\ast)^2}{\eps})}\, \D E + C |\Delta| \e^{-\frac{M(E^\ast)}{\eps}} \eps^m 
\nonumber\\
&\leq& C  \e^{-M(E^\ast)/\eps}\eps^{1/2}     .
\eea
Similarly, $\norm[L^2]{Q_0} \leq C \e^{-M(E^\ast)/\eps} \eps^{1/4}$. This shows the claim.
\end{proof}

Formulas (\ref{q sharp}) and (\ref{L1 norm}) are also the key to the following

\begin{proposition} \label{concentrated}
Assume that $V$ fulfills (V1)--(V4)$_m$, and that $Q$ fulfils (C1)--(C3). 
Put 
\[
\Delta_{\eps} = \{ E \in \Delta: |E-E^{\ast}|<\eps^{1/2-\delta} \},
\]
and 
\be \label{chiL0}
\chi_{L,0}(\eps,x,t):=\int_{\Delta_\eps}Q(E,\varepsilon)e^{-itE/\varepsilon} 
c_{2,L}(\eps,x,E)\ \varphi_{\eps,2}^{(1)}(\eps,x,E) dE.
\ee
Then for all $m > 0$ there exist $\eps_0>0$ and $C>0$ such that for all $\eps < \eps_0$ we have 
\[
\norm[L^2]{\chi_{L0}(\eps,. \, , t) - \chi_L(\eps,. \, , t) } \leq  C  \E^{-\frac{M(E^{\ast})}{\eps}} \eps^{m}.
\]
\end{proposition}
\begin{proof}
Note that by (\ref{35a}) we have $\varphi_{n,2}^{(1)}(\eps,x,E) = \frac{1}{\sqrt{p(x,E)}} + r_{1}$ with $r_{1} \in \ec_{1}(1,1+\nu)$; moreover, we have 
\be \label{erf estimate}
1 - \mr{erf}(w)  = \frac{1}{\sqrt{\pi}} \int_w^\infty \e^{-u^2} \, \D u \leq \frac{1}{\sqrt{\pi}}
\int_w^\infty \frac{u}{w} \e^{-u^2} \, \D u = \frac{2}{\sqrt{\pi} w} \e^{-w^2}
\ee
for $w>0$ and $\mr{erf}(w)\leq 2e^{-w^2}/(\pi |w|)$ for $w<0$.
This shows 
$c_{2,L}(\eps,x,E)\ \varphi_{\eps,2}^{(1)}(\eps,x,E) = \e^{-\frac{\tc}{\eps}}(r_1 + r_2)$
with $r_1 \in \ec_1(0,1)$ and $r_2 \in \ec_2(0)$. Then (\ref{q sharp}) and Lemma \ref{L00} shows that 
\[
\norm[L^1] {\int_{\Delta \setminus \Delta_\eps}Q(E,\varepsilon)e^{-itE/\varepsilon} 
c_{2,L}(\eps,x,E)\ \varphi_{\eps,2}^{(1)}(\eps,x,E) dE} \leq C \e^{-\frac{M(E^\ast)}{\eps}} \eps^m
\]
for any $m$, the same holds for the corresponding $L^2$ norm, and the claim follows.
\end{proof}

The final step in the proof of Theorem \ref{chi eff thm} is an effective description of the error function in $c_{2,L}$. Define $x_{\mr r}(E)$ and $a(x,E)$ as in (\ref{a(x,E)}) and note that $E \mapsto x_{\mr r}(E)$ is $C^1$ by (V3). Put
\[
j(x,E) = \left\{ \begin{array}{ll}
1 & \mbox{if } x > x_{\mr r}(E), \\
0 & \mbox{if } x \leq x_{\mr r}(E).
\end{array} \right.
\]
\begin{proposition} \label{erf prop}
For any $m>0$, $k>0$ and $\delta > 0$ we have 
\[
\sup_{E \in \Delta_\eps} \left| \mr{erf}(a(x,E)/\sqrt{\eps}) - j(x,E^\ast)\right| 1_{\{|x-x_{\mr r}(E^\ast)| > \eps^{1/2-2 \delta}\}} \in \ec_1(m,k).
\]
\end{proposition}
\begin{proof}
Let us consider only the case $x > x_{\mr r}(E^\ast)$, the other case goes analogously. 
We first show the estimate $x-x_r(E)=(x-x_r(E^*))(1+O(\eps^\delta))$.
By (V3) we have $|x_{\mr r}(E) - x_{\mr r}(E^\ast)| \leq C|E-E^\ast|$ for some $C>0$.
Thus,
\[
 x-x_r(E)=(x - x_{\mr r}(E^\ast) - C|E^\ast-E|) \geq  |x - x_{\mr r}(E^\ast)| -  C \eps^{1/2 - \delta}
\]
on $\Delta_\eps$, for another constant $C>0$. 
Since $|x- x_{\mr r}(E^\ast)| > \eps^{1/2 - 2 \delta}$ implies 
$\eps^{1/2 - \delta} < |x- x_{\mr r}(E^\ast)| \eps^\delta$, we find 
\[
x-x_r(E)=  (x - x_{\mr r}(E^\ast))(1 - C \eps^{\delta}).
\]
By 
(\ref{erf estimate}) we find 
\be \label{eq001}
j(x,E^\ast) - \mr{erf}(a(x,E) / \sqrt{\eps}) = 1 -  \mr{erf}(a(x,E) / \sqrt{\eps}) 
\leq \frac{2 \sqrt{\eps}}{\sqrt{\pi} a(x,E)} \e^{-\frac{a^2(x,E)}{\eps}}.
\ee
Now since $p(y,E) > c$ for all $y$ and all $E$, and since $\tc(E)$ is bounded away from $0$, we get 
\[
a(x,E)  = \frac{1}{\sqrt{2 \tc(E)}} \int_{x^\ast_{\mr r}}^x p(y,E) \, \D y \geq c_0 |x-x_{\mr r}(E)|.
\]
From the estimate above we find
\[
a(x,E) \geq c_0 (x - x_{\mr r}(E^\ast))(1 - C \eps^{\delta})
\]
on $\Delta_\eps \times \{ x: |x- x_{\mr r}(E^\ast)| > \eps^{1/2 - 2 \delta} \}$. 
Plugging this into (\ref{eq001}) and using Lemma \ref{gaussfact}, we obtain the result.
\end{proof}

Now Theorem \ref{chi eff thm} follows by using Proposition \ref{erf prop} in the expression for 
$c_{2,L}$ in Proposition \ref{concentrated} in case $|x-x^\ast_{\mr r}| > \eps^{1/2 - \delta}$, 
while keeping $c_{2,L}$ for $|x-x^\ast_{\mr r}| \leq \eps^{1/2 - \delta}$, and then using Lemma \ref{L00} again.

We now turn to the proof of Theorem \ref{reflected wave shape}. Let us recall the abbreviations 
\[ 
M(E) = \tc(E) + G(E), \quad S(x,E,t) =  \tr(E) + J(E) - \frac{1}{2} \xi(x,E) + Et, 
\]
and 
\[
a(x,E)=\frac{\xi(x,E)-\xi_r(E)}{\sqrt{2\xi_c(E)}}=\frac{\int_{x_r(E)}^{x}p(y,E)\,dy}{\sqrt{2\xi_c(E)}}. 
\]
We first take advantage of Proposition \ref{concentrated}, i.e.\ the fact that we know $E-E^{\ast}$ is small in the relevant area of $\Delta$, and expand all quantities of $E$ around $E^{\ast}$. The only exception is $\xi(x,E)$ since $x \mapsto \xi(x,E)$ grows like $p(\infty,E) x$, and so in this case such an expansion would not be uniform in $x$. We have 
\be
\frac{a(x,E)}{\sqrt{\eps}}=\frac{a(x,E^\ast)}{\sqrt{\eps}}+\frac{\partial_E a(x,E^\ast)}{\sqrt{\eps}}(E-E^\ast)
+\Or \left(\frac{(E-E^\ast)^2}{\sqrt\eps}\right), 
\ee
where the error term is not uniform in $x$. By the Taylor formula with integral remainder 
\be\label{tr}
\mr{erf}(z+z_{0})=\mr{erf}(z)+\frac{z_{0}}{\sqrt{\pi}}\int_0^1 e^{-(z+sz_{0})^2}\,ds,
\ee
we can write
\be
 \mr{erf}\left(\frac{a(x,E)}{\sqrt{\eps}}\right)= \mr{erf}\left(\frac{a(x,E^*)}{\sqrt{\eps}}+\frac{\partial_E a(x,E^*)}{\sqrt{\eps}}(E-E^*)\right)+O\left(\frac{(E-E^*)^2}{\sqrt\eps}\right).
\ee
The error term is uniform in $x$ here, because the error functions are. Since we are on $\Delta_{\eps}$, by restricting attention to the set $\{|x-x^\ast_{\mr r}|\leq \eps^{1/2 - \delta}\}$ we actually get
\be
O\left(\frac{(E-E^*)^2}{\sqrt\eps}\right)1_{\{|x-x^\ast_{\mr r}|\leq \eps^{1/2 - \delta}\}}\in \ec_1(\eps^{1/2-2 \delta},\infty).
\ee
Hence, by (\ref{L1 norm}) and Lemma \ref{L00} again,  we find that once multiplied by $Q_0(E)$ and the , the $L^{2}(dx)$-norm of the corresponding $\D E$-integral of this remainder is of order $\Or(\E^{-\frac{M(E^{\ast})}{\eps}} \eps^{1-3\delta})$. 
Similarly, we expand 
\[
S_{0}(E,t) := + \xi_{r}(E) + J(E) + Et
\]
and $M(E)$ to third order. The error stemming from this expansion is of order
$|E-E^\ast|^3/\eps$ and is uniform in $t$ and $x$. Thus, by the same argument,
this remainder eventually yield an error of order $\Or(\E^{-\frac{M(E^{\ast})}{\eps}} \eps^{1-4\delta})$ in $L^{2}(dx)$-norm on the set  $\{|x-x^\ast_{\mr r}|\leq \eps^{1/2 - \delta}\}$.

Finally, we use (\ref{35a}) to see that $\varphi_{n,2}^{(1)}(\eps,x,E) = \frac{1}{\sqrt{p(x,E)}} + r_{1}$ with $r_{1} \in \ec_{1}(1,1+\nu)$ and estimate the error term using the  by now familiar reasoning.  Defining 
\[ P_{0}(x,E) := \frac{ 2 P(0,E) \sin\left( \pi\gamma / 2 \right)}{\sqrt{p(x,E)}}
\]
and expanding around $E^*$ we obtain 
$$
 P_{0}(x,E)= P_{0}(x,E^\ast)+O(E-E^\ast),
$$
with an error term that is uniform in $x$ and of order $\eps^{1/2-\delta}$, so that we have shown 
\begin{proposition} \label{P05}
Define 
\be \label{xxx}
\begin{array}{l}\chi_{L,1}(\eps,x,t) =  e^{-(M(E^{\ast})+iS_0(E^{\ast},x,t))/\eps}\int_{\Delta_{\eps}} 
P_0(x,E^{\ast})
\E^{-\frac{\I (E-E^\ast)}{\eps} S_{0}'(E^\ast,t)} \times \nonumber \\[5mm]
\times \E^{-\frac{(E-E^\ast)^2}{2 \eps}(M''(E^\ast)+iS_{0}''(E^\ast,t))} \E^{\frac{\I}{2 \eps} \xi(x,E)}
\mr{erf}\left(\frac{a(x,E^*)}{\sqrt{\eps}}+
\frac{a'(x,E^\ast)}{\sqrt{\eps}}(E-E^\ast)\right) \, \D E,
\end{array}
\ee
where all the primes denote $E$-derivatives. 
Then, for all $0<\delta<1/4$, 
$$\norm[L^{2}(\D x)]{(\chi_{L,1} - \chi_{L,0})1_{\{ |x-x^\ast_{\mr r}|\leq \eps^{1/2 - \delta} \}}} = \Or(\E^{-M(E^{\ast})/\eps} \eps^{1-4\delta}).$$ 
\end{proposition}

To get the result stated in Theorem \ref{reflected wave shape} on the range
$\{ |x-x^\ast_{\mr r}|\leq \eps^{1/2 - \delta} \}$, we need to expand the 
term $\xi(x,E)$ around $E^{\ast}$ as
\be \label{xi expand}
\E^{\frac{\I}{2\eps}\xi(x,E)} = \E^{\frac{\I}{2\eps} ( \xi(x,E^{\ast}) + \xi'(x,E^{\ast})(E-E^{\ast}) + \frac{1}{2}\xi''(x,E^{\ast})(E-E^{\ast})^{2})}\left(1 + \Or\left(\frac{(E-E^{\ast})^{3}}{ \eps} x \right) \right),
\ee
where primes denote $E$-derivatives. Over the considered set of $x'$s, the error term generated is dominated by $|E-E^\ast|^3/\eps$, and is dealt with as above. To arrive at the expression $\chi_{\mr near}$, we change variables from $E$ to $y=(E-E^\ast)/\sqrt\eps$ in the remaining integral, and  note that extending the range of integration to the whole of ${\mathbb R}$ costs an error of order $\eps^\infty \E^{-M(E^{\ast})/\eps}$.

We now 
consider the region $\eps^{1/2-\delta}<|x-x_r^\ast| < \eps^{-\beta}$ with $\beta < 1/8$ and $0<\delta<1/4$. The following proposition  settles Theorem \ref{reflected wave shape} for the intermediate region.

\begin{proposition} \label{P06} 
Define $\chi_{\mr{mod}}$ as in (\ref{chimod}).
Then there exist $\eps_0 > 0$ and $C>0$ such that for each $\eps < \eps_0$ we have 
\[ 
\norm[L^{2}(\D x)]{
(\chi_{\mr{mod}} - \chi_{L,1}) 1_{\{\eps^{1/2 - \delta} < x-x^\ast_{\mr r} <\eps^{-\beta}\}}} \leq C \E^{-\frac{M(E^{\ast})}{\eps}} \eps^{1-4 \delta-2\beta}.
\]
\end{proposition}
\begin{proof}
The leading term is obtained from the leading term of (\ref{xi expand}): we know from Proposition 
\ref{erf prop} that we can replace the error function in (\ref{xxx}) with $1$ for the range of $x$ we consider, and then we obtain a purely Gaussian integral in $(E-E^\ast)$. Again, we can extend the range of integration to all or $\R$ for the cost of 
an error of $\Or(\e^{-M(E^{\ast})/\eps} \eps^{\infty})$. 
When integrated over $\eps^{1/2 - \delta} < x - x^\ast_{\mr r} <\eps^{-\beta}$ this error stays of that order, and distributing the various leading terms of (\ref{xi expand}) in the exponent changes the $S_{0}$ and its derivatives to $S$. 
For the error term we note that $|E-E^{\ast}| < \eps^{1/2-\delta}$ on $\Delta_{\eps}$, and thus 
\[
\Or \left( \frac{|E-E^{\ast}|^{3} |x| }{\eps} \right) \leq \Or(\eps^{1/2 - 3 \delta - \beta})
\]
for $|x| < \eps^{-\beta}$. Since the length of the range of integration  in $x$ is $\Or(\eps^{-\beta})$, the length of integration range in $E$ is $\eps^{1/2 - \delta}$, and the factor $\E^{-\frac{M(E^{\ast})}{\eps}}$ is present, we obtain the result.
\end{proof}

It remains to treat the region $x>\eps^{-\beta}$. Since this has essentially been done before \cite{jm05}, we will be much briefer here. We use Lemma \ref{L01} in order to write
\[
1_{\{x>\eps^{-\beta}\}} \E^{\frac{\I}{2\eps} \xi(x,E)} = \E^{\frac{\I}{\eps} (\om(E) + p(\infty,E) x)} (1+r_{1}),
\]
with $r_{1} \in \ec_{1}(\beta(m-\nu)-1,\nu)$.  By Proposition \ref{erf prop}, we can replace the error function by $1$, and replace $p(x,E)$ by $p(\infty,E)$ with an error bounded by $C/\bk{x}^{m+3/2}$ and thus by $C \eps^{\beta(1+m-\nu)}\bk{x}^{\nu+1/2}<<C \eps^{\beta(m-\nu)-1}\bk{x}^{\nu+1/2}$. Integrating these error terms over $E$ against the Gaussian function gives a term whose $L^2(\D x)$ norm is bounded by means of Lemma \ref{L00}. The  $L^1(\D E)$ norm  of the Gaussian in energy is of order
$\eps^{1/2}$, see (\ref{L1 norm}), so that the error term is eventually of $L^{2}(\D x)$ norm of order $\Or(\E^{-\frac{M(E^{\ast})}{\eps}} \eps^{\beta(m-\nu) - 1/2})$. 
We need $\beta(m-\nu)-1/2 > 3/4$, so $m>5/(4\beta) > 10$ with $\beta < 1/8$. Finally, the leading term can be computed by Gaussian integration. 
First we change variables
from $E\mapsto p(\infty,E)=\sqrt{E-V(\infty)}$ to 
$k\mapsto E(k)=k^2+V(\infty)$ and write $\tilde f (k)=f(E(k))$ for a function $f$, just as we did in the proof of Lemma \ref{L01}.
Carrying out the resulting Gaussian integration, we find that if $m>10$, then there exists $\delta>0$ such that 
\[ \chi_{L,0}(\eps,x,t) = \chi_{\mr{far}}(\eps,x,t) + \Or(\E^{-\frac{M(E^{\ast})}{\eps}} \eps^{3/4 + \delta}),\]
where $\chi_{\mr{far}}$ is given by (\ref{xfar}). This finishes the proof of Theorem \ref{reflected wave shape}.\\

It remains to prove Theorem \ref{chimod gauss}. The idea of the proof is to observe that $\chi_{\mr{near}}$ and $\chi_{\mr{mod}}$ are sharply peaked around the classical trajectory $q_t$ of the reflected wave, i.e.\ the solution unique solution of $S'(E^\ast,q_t,t) = 0$. This will allow us to expand the reflected wave around its maximum in the variable $x$, and get the result. Since 
we already know that $\chi_{\mr{far}}$ is a Gaussian centred around the asmptotically freely moving classical reflected wave, \cite{jm05}, it will be negligible for the times $t$ considered. We start by proving concentration in $x$ for finite times. Recall that by the change of variables 
$y = (E-E^\ast)/\sqrt\eps$, we can write, for any $D\in{\mathbb R}$
\begin{eqnarray}
\lefteqn{
\chi_{L,1} (\eps,x,t) = 1_{\{x>D\}}\sqrt{\eps} P_0(\eps,x,E^{\ast})   
\e^{-\frac{M(E^\ast)+iS(E^{\ast},x,t)}{\eps}} \times} \label{chiL1 with y}\\ 
&& \times  
\int_{|y| < \eps^{-\delta}}
\mr{erf}\left(\frac{a(x,E^*)}{\sqrt{\eps}}+a'(x,E^\ast)y\right)
\e^{-y^{2} \frac{M''(E^\ast)+iS''(E^\ast,x,t)}{2}} \e^{-\frac{\I}{\sqrt{\eps}} S'(E^{\ast},x,t) y}
\,dy. \nonumber
\end{eqnarray}

\begin{proposition} \label{concentration x finite}
For all $m>0$, all $\delta > 0$ and all $\tilde C > 0$ there exists $C>0$ such that 
\[
\norm[L^{2}(\D x)]{\chi_{L,1} 1_{\{ |x-q_{t}| > \eps^{1/2-\delta}, |x| < \tilde C \}} } 
\leq C \e^{-\frac{M(E^{\ast})}{\eps}} \eps^{m}
\]
for all $t \in \R$. 
\end{proposition}

\begin{proof}
Let us first note that we can extend the range of integration to all $y \in \R$ in (\ref{chiL1 with y}) at the expense of an error bounded by $C \e^{-\frac{M(E^{\ast}}{\eps}} \eps^{m}$ in $L^2(\D x)$ sense, by the usual arguments. 
For $\delta>0$ choose $k \in \N$ such that $k>m/\delta$. $k$-fold integration by parts gives 
\begin{eqnarray*} 
\chi_{L,1} (\eps,x,t) &=& 1_{\{x>D\}}\sqrt{\eps} P_0(\eps,x,E^{\ast}) (-1)^{k}  
\e^{-\frac{M(E^\ast)+iS(E^{\ast},x,t)}{\eps}} \times \\ 
&& \times \int_{\mathbb R} \frac{\D^{k}}{\D y^{k}} \left( 
\mr{erf}\left(\frac{a(x,E^*)}{\sqrt{\eps}}+a'(x,E^\ast)y\right)
\e^{-y^{2} (M''(E^\ast)+iS''(E^\ast,x,t))/2}\right) \times \\
&& \times \frac{\eps^{k/2}}{(-i S'(E^{\ast},x,t))^{k}} \e^{-\frac{\I}{\sqrt{\eps}} S'(E^{\ast},x,t) y}
\,dy + \Or\left(  \e^{-\frac{M(E^{\ast}}{\eps}} \eps^{m} \right). 
\end{eqnarray*}
The $k$-fold derivative has the form $f(\eps,x,y) \exp(-y^2 M''(E^\ast)/2)$, where 
\be \label{bbb}
|f(\eps,x,y)| \leq C(k) |y|^k
\ee
  uniformly in $\eps > 0$, in time $t$ and $|x| < \tilde C$. To see this, note that by the formula $\partial_x^l g(a + b x) = b^n g^{(l)}(a + bx)$, the $l$-th derivative of the error 
  function is an Hermite function, i.e.\ is given by 
  \[
  a'(x,E^\ast)^l H_{l-1}\left(a(x,E^*)/\sqrt{\eps} + a'(x,E^\ast) y\right) e^{- (a(x,E^*)/\sqrt{\eps} + a'(x,E^\ast) y)^2},
 \]
   where $H_l$ is the $l$-th Hermite polynomial. 
Since Hermite functions are uniformly bounded in their variable, the uniformity claim in $\eps$ follows. Since the derivatives on the second factor also produce Hermite functions, we have shown 
(\ref{bbb}). 
Moreover, $|S'(E^{\ast},x,t)| > c |q_{t}-x|$ for some $c>0$ uniform in $t$, and so  
\[ 
1_{\{|q_{t}-x| > \eps^{1/2-\delta}\}} \frac{\eps^{k/2}}{|S'(E^{\ast},x,t)|^{k}} \leq \eps^{k \delta}/c^k \leq \eps^{m}/c^k.
\]
Thus the $\D y$ integral is bounded by $C(\delta,m) \eps^{m}$ and bounded uniformly in $x$ and $t$, and estimating the $x$ integral by the area $|x|<\tilde C$ times its maximum gives the result.
\end{proof}

In the region where $x$ is moderately large but not finite as $\eps\rightarrow 0$, we need a different argument due to the broadening of the wave packet. 

\begin{proposition} \label{concentration x moderate}
Put 
\[
\Lambda = \{ x \in \R: \eps^{1/2 - \delta} \sqrt{M''(E^\ast)^2 + 2 S''(E^\ast,q_t,t)^2} < |x-q_t| < \eps^{- \beta} \}.
\]
Then for each $m \in \N$ we have 
\[
\norm[L^2(\D x)]{\chi_{\mr{mod}} 1_{\Lambda}} \leq K \e^{-\frac{M(E^\ast)}{\eps}} \eps^m , \ \ \ \mbox{uniformly in $t$.}
\]
\end{proposition}
\begin{proof}

First recall that $ S''(E^\ast,x,t)$ is independent of $t$ and grows with $x$ like $x$. The negative of the real part of the exponent of $\chi_{\mr mod}$ is given by 
\[
f(x) := \frac{S'(E^\ast,x,t)^2 M''(E^\ast)}{2 \eps (M''(E^\ast)^2 + S''(E^\ast,x,t)^2)}.
\]
Since $|S'(E^\ast,x,t)| \geq c|q_t-x|$ for some $c>0$ and $|S"(E^*,x)-S"(E^*,q_t)|\leq C|q_t-x|$  for some $C>0$, we find 
\be \label{f>}
f(x) \geq \frac{c^2 |x-q_t|^2 M''(E^\ast)}{2 \eps ( M''(E^\ast)^2 + 2 S''(E^\ast,q_t,t)^2 +
 2 C |q_t - x|^2)}.
\ee
Differentiation of the function $y \mapsto \frac{y}{a^2 + y b^2}$ shows that the right hand side of (\ref{f>}) is minimal when $|x-q_t|^2$ is minimal, i.e.\ when $|x-q_t|$ 
is equal to its lower bound in $\Lambda$. Inserting this yields
\[
f(x) \geq \eps^{-2 \delta} \frac{c^2( M''(E^\ast)^2 + 2 S''(E^\ast,q_t,t)^2) M''(E^\ast)}{2 ( M''(E^\ast)^2 + 2 S''(E^\ast,q_t,t)^2)(1 + 2C\eps^{1-2\delta})}=\eps^{-2 \delta} \frac{c^2 M''(E^\ast)}{2 (1 + 2C\eps^{1-2\delta})}.
\]
This shows that $|\chi_{\mr{mod}}| \leq C e^{-\frac{M(E^\ast)}{\eps}} \eps^m$ for each $m \in \N$ all $x \in \Lambda$, with $C$ uniform in $t$. Now, since for $t$ large $q_t\simeq t$ so that $ S''(E^\ast,q_t,t)\simeq t$, 
we get that on $\Lambda$, $t\leq c_1 \eps^{-\beta-1/2+\delta}$, so that $|x|\leq c_2 \eps^{-\beta-1/2+\delta}$, for 
some constants $c_1, c_2$. Integrating over the range $|x| < c_2\eps^{- \beta-1/2+\delta}$ then shows the claim.
\end{proof}

We can now complete the proof of Theorem \ref{chimod gauss}. Define $\chi_{\mr{gauss}}$ as in Theorem \ref{chimod gauss}. Put $\Lambda_- = \{x: x - x^\ast_{\mr r} < \frac{1}{2} \eps^{1/2 - \delta} \}$, and 
$\Lambda_+ = \{x: x - x^\ast_{\mr r} > 2 \eps^{- \beta} \}$
\begin{eqnarray*}
\norm{\chi_{\mr{gauss}}-\chi_{\mr{eff}}} & \leq & \norm{\chi_{\mr{gauss}} (1_{\Lambda_+} + 1_{\Lambda_-})} + \norm{\chi_{\mr{near}} 1_{\Lambda_-}} + \norm{\chi_{\mr{far}} 1_{\Lambda_+}} + \\
&& +  \norm{(\chi_{\mr{gauss}} - \chi_{\mr{mod}}) 1_{\{ \eps^{1/2 - \delta}/2 < x - x^\ast_{\mr r} < 2 
\eps^{-\beta}\}}} + \Or\left( \e^{-\frac{M(E^\ast)}{\eps}} \eps^{3/4 + \delta_0} \right),
\end{eqnarray*}
where all norms are $L^2$ norms, and $x^\ast_{\mr r} = x_{\mr r}(E^\ast)$. Now for $\eps^{1/2 - \delta} < |q_t-x_r^\ast|  < \eps^{-\beta}$, the first and third term are $\Or\left( \e^{-\frac{M(E^\ast)}{\eps}} \eps^m \right)$ for any $m>0$, as can be seen by direct estimates on Gaussian tails. By Proposition 
\ref{concentration x finite} the second term is of the same order for $|q_t-x_r^\ast| > \eps^{1/2 - \delta}$. For the term in the second line, we may restrict our attention to the region $|x-q_t| < \eps^{1/2 - \delta} \sqrt{M''(E^\ast)^2 + 2 S''(E^\ast,q_t,t)^2}$. 

This follows from Proposition \ref{concentration x moderate} for $\chi_{\mr{mod}}$, and from direct Gaussian estimates for $\chi_{\mr{gauss}}$. For these $x$ we expand the exponent of $\chi_{\mr{mod}}$; we have to investigate 
\[
d(x,t) := \left| \frac{S'(E^{\ast},x,t)^2}{\eps (M''(E^\ast) + \I S''(E^\ast,x,t))} - \frac{(x-q_t)^2}{4 \eps p(E^\ast,q_t)^2 (M''(E^\ast) + \I S''(E^\ast,q_t,t))} \right|. 
\]
Since both  $|\partial_x S'(E^\ast,x,t)| \leq C $ and $|\partial_x S''(E^\ast,x,t)| \leq C $ uniformly in $x$ and $t$, we have
\[
\left| S'(E^\ast,x,t)^2 - \frac{(x-q_t)^2}{4 p^2(E^\ast,q_t)} \right| \leq C |x-q_t|^3, 
\]
and 
\bea
&&\left| 
 \frac{1}{(M''(E^\ast) + \I S''(E^\ast,x,t))} -  \frac{1}{(M''(E^\ast) + \I S''(E^\ast,q_t,t))} \right| \nonumber\\ \nonumber
&&\quad\quad\leq C \frac{|x-q_t|}{|M''(E^\ast) + \I S''(E^\ast,x,t)||M''(E^\ast) + \I S''(E^\ast,q_t,t)|}.
\eea
Then we plug this into the expression for $d(x,t)$, and make use of the inequality $|ab-cd| \leq |a-c||b| + |b-d||c|$, and of the estimate
\bea
\lefteqn{ M''(E^\ast) + \I S''(E^\ast,x,t) = } \nonumber \\  & = & (M''(E^\ast) + \I S''(E^\ast,q_t,t))
\left(1+\frac{i(S''(E^\ast,x,t)-S''(E^\ast,q_t,t))}{M''(E^\ast) + \I S''(E^\ast,q_t,t)}\right)\nonumber\\ \label{lastest}
& = &(M''(E^\ast) + \I S''(E^\ast,q_t,t))\left(1+O\left(\frac{|x-q_t|}{|M''(E^\ast) + \I S''(E^\ast,q_t,t)|}\right)\right).
\eea
On the set $|x-q_t|\leq \sqrt{2}\eps^{1/2-\delta}|M''(E^\ast) + \I S''(E^\ast,q_t,t)|$, the last term of the parenthesis of  order $\eps^{1/2-\delta}$, uniformly in $x,t$.
Eventually, we get on the set 
\[ |x-q_t|\leq \eps^{1/2-\delta}\sqrt{M''(E^\ast)^2+2S''(E^\ast,q_t, t)}, \] for uniform constants $C_1$, $C_2$, 
\[
d(x,t) \leq \frac{C_1}{\eps} |x-q_t|^3 \leq C_2 \eps^{1 - 3 \delta}.
\]
It then remains to make use of these estimates in the computation of the 
modulus of $\chi_{\mr{gauss}} - \chi_{\mr{mod}}$ on the set  $|x-q_t|\leq \eps^{1/2-\delta}\sqrt{M''(E^\ast)^2+2S''(E^\ast,q_t, t)}$. We finally get, with a constant $C_3$ uniform   in $t$,
\[
|(\chi_{\mr{gauss}} - \chi_{\mr{mod}}) 1_{\{ \eps^{1/2 - \delta}/2 < x - x^\ast_{\mr r} < 2 
\eps^{-\beta}\}}| \leq C_3|\chi_{\mr{gauss}}| \eps^{1 - 3 \delta}.
\]
Explicit integration of the right hand side above finishes the proof. 

Let us finally remark that in a similar way, one can expand the expression (\ref{chiL1 with y}) for the near region, and get an approximate expression for the transition region as well.

\section{Superadiabatic representations} \label{recursion}
In this section we will prove Theorem \ref{n-th}. We start with equation (\ref{ODE2}) which we repeat here for convenience: 
\begin{equation} \label{ODE2a}
\I \eps \partial_{\xi} \psi(\xi) =  H(\xi) \psi(\xi) := \frac{1}{2}\left( \begin{array}{cc} 0 & \frac{1}{\tilde p(\xi)}\\ \tilde p(\xi) & 0 \end{array}\right) \psi(\xi).
\end{equation}
While the ultimate aim is to construct the superadiabatic transformations $T_{n}$ from theorem \ref{n-th}, we start by constructing superadiabatic {\em projections}. For this matter, write 
\[
\pi_+(\xi) =  \frac{1}{2}\left( \begin{array}{cc} 1 & \frac{1}{ \tilde p(\xi)}\\ \tilde p(\xi) & 1 \end{array}\right)\,, \quad \pi_-(\xi) = {\bf 1}- \pi_+(\xi)\,.
\]
for the spectral projections corresponding to the eigenvaluse 
$\pm 1/2$ of $H$. The $n$-th superadiabatic projection will be a matrix 
\begin{equation} \label{ansatz}
\pi^{(n)} = \sum_{k=0}^n\varepsilon^k \pi_k \,,
\end{equation}
with $\pi_0(x) = \pi_+(x)$ and 
\begin{eqnarray}
 & & (\pi^{(n)})^2 - \pi^{(n)} = \Or(\varepsilon^{n+1}),  \label{pi1}\\
 & & \commut{\I\eps \partial_\xi - H}{\pi^{(n)}} = \Or(\varepsilon^{n+1}). \label{pi2}
\end{eqnarray}
Here, $\commut{A}{B} = AB - BA$ denotes the commutator two
operators $A$ and $B$. Likewise, we will later use $\antcom{A}{B}
= AB + BA$ to denote the anti-commutator of $A$ and $B$. Equation
(\ref{pi1}) says that $\pi^{(n)}$ is a projection up to errors of
order $\eps^{n+1}$, while (\ref{pi2}) states that if $\xi \mapsto \psi(\xi) \in \C^{2}$ solves (\ref{ODE2a}), then also  $\xi \mapsto \pi^{(n)}(\xi) \psi(\xi)$ solves (\ref{ODE2a}) up to errors of order 
$\eps^{n+1}$. 

We construct $\pi^{(n)}$ inductively starting from the Ansatz (\ref{ansatz}).
Obviously  (\ref{pi1}) and (\ref{pi2}) are fulfilled
for $n=0$. In order to construct $\pi_n$ for $n > 0$, let us write
$G_n(t)$ for the term of order $\eps^{n+1}$ in (\ref{pi1}), i.e.
\begin{equation} \label{G def}
(\pi^{(n)})^2 - \pi^{(n)} = \eps^{n+1} G_{n+1} +
\Or(\eps^{n+2})\,,
\end{equation}
with 
\begin{equation} \label{G calculated}
G_{n+1} = \sum_{j=1}^{n} \pi_j \pi_{n+1-j}.
\end{equation}

The following proposition is taken from \cite{BeTe1} and the proof there applies literally also to the present case.
\begin{proposition} \label{general scheme}
Assume that $\pi^{(n)}$ given by (\ref{ansatz}) fulfills
(\ref{pi1}) and (\ref{pi2}). Then a unique matrix $\pi_{n+1}$
exists such that $\pi^{(n+1)}$ defined as in (\ref{ansatz})
fulfills (\ref{pi1}) and (\ref{pi2}). $\pi_{n+1}$ is given by
\begin{equation}\label{recursion1}
\pi_{n+1} = G_{n+1} - \antcom{G_{n+1}}{\pi_0} - \I
\commut{\pi'_n}{\pi_0}.
\end{equation}
Furthermore $\pi'_n$ is off-diagonal with respect to $\pi_0$, i.e.
\begin{equation} \label{offdiag}
\pi_0  \pi'_n \pi_0  =  (1-\pi_0) \pi'_n (1-\pi_0) = 0,
\end{equation}
and $G_{n}$ is diagonal with respect to $\pi_0$, i.e.
\begin{equation} \label{diag}
\pi_0 G_{n+1} (1-\pi_0)  =  (1-\pi_0) G_{n+1} \pi_0 = 0.
\end{equation}
\end{proposition}

As in \cite{BeTe1} we represent the
  matrix recurrence relation
(\ref{recursion1})  with respect to a special basis of
$\C^{2 \times 2}$, in order to arrive at a more tractable set of scalar recurrence relations for the coefficients. While the following matrices are different from the one chosen in \cite{BeTe1}, they satisfy algebraic relations very close to those obtained and utilized in \cite{BeTe1}. As a consequence most of our proofs require only slight adjustments as compared to those in \cite{BeTe1, BeTe2}.
Let
\[
 X = \left(\begin{array}{cc}-1 & 0 \\ 0 & 1 \end{array}\right),  \quad Y = -2 H, \quad Z = \frac{-1}{\theta'}  Y'\,,
\]
where $\theta' = \tilde p'/ \tilde p$. Together with the identity matrix $\bf 1$, $X$, $Y(\xi)$ and $Z(\xi)$ from a basis of $\R^{2 \times 2}$ for every $\xi\in\R$.
The following  relations now follow without effort:
\begin{eqnarray}
&&  X' = 0, \quad  Y' = -  \theta' Z, \quad  Z' =  -\theta' Y, \label{derivatives} \\
&& \antcom XY = \antcom XZ = \antcom YZ = 0,\quad  X^2 = Y^2 = -Z^2 = {\bf 1}, \label{products} \\
&& \commut{X}{\pi_0} = Z, \quad \commut{Y}{\pi_0} = 0, \quad  \commut{Z}{\pi_0} = X, \label{commutators} \\
&& {\bf 1} - \antcom{{\bf 1}}{\pi_0} = Y. \label{W equation}
\end{eqnarray}
The relations (\ref{derivatives})--(\ref{W equation}) agree almost exactly with the corresponding relations derived and used in \cite{BeTe1}. The only differences are a sign change in the last equation in (\ref{derivatives}) and in the last equation in (\ref{products}). 
These relations show that this basis behaves extremely well under
the operations involved in the recursion (\ref{recursion1}). This
enables us to obtain

\begin{proposition} \label{function recursion}
For all $n \in \N$, $\pi_n$ is of the form
\begin{equation} \label{xyz ansatz}
\pi_n = x_nX + y_nY + z_nZ,
\end{equation}
where the functions $x_n, y_n $ and $z_n$ solve the scalar recursion
\begin{eqnarray}
 x_{n+1} & = &  -\I z_n' +\I \theta' y_n\label{xRec}\\
 y_{n+1} & = &  \sum_{j=1}^n \left( x_j x_{n+1-j} +y_j y_{n+1-j}    -z_j z_{n+1-j}   \right)\label{yRec}\\
 z_{n+1} & = &  -\I x_n'\,,\label{zRec}
\end{eqnarray}
and satisfy the differential
equations
\begin{eqnarray}
 x'_n & = &  \I  z_{n+1}, \label{E1a}\\
 y'_n & = &  \theta' z_n, \label{E1b}\\
z'_n  & = &   \I  x_{n+1} + \theta'  y_n.  \label{E1c}
\end{eqnarray}
Moreover,
\begin{equation} \label{functions recursion start}
 x_1(\xi) = \frac{\I}{2}\theta'(\xi), \quad y_1(\xi) = z_1(\xi) = 0.
\end{equation}
\end{proposition}

\begin{remark}
Hence, for all even $n$, $x_n = 0$, while for all odd $n$, $y_n =
z_n = 0$. Moreover, $x_n$ is always purely imaginary, while $y_n$ and $z_n$ are real. 
\end{remark}

\begin{proof} For $n=0$ we have $G_1=0$ and thus 
\[
\pi_1 = -\I [\pi_0',\pi_0] = \frac{\I}{2}\theta' [Z,\pi_0] = \frac{\I}{2}\theta' X\,,
\]
which shows (\ref{functions recursion start}). For the rest we just have to stick (\ref{xyz ansatz}) into the recurrence (\ref{recursion1}) and use the relations (\ref{derivatives})--(\ref{W equation}). For $G_{n+1}$ we use (\ref{products}) to obtain
\begin{eqnarray*}
G_{n+1} &=& \sum_{j=1}^n \left(x_jX+ y_jY+z_jZ  \right)\left( x_{n+1-j}X+ y_{n+1-j}Y+z_{n+1-j}Z  \right)\\&=&
 \sum_{j=1}^n \left( x_j x_{n+1-j} + y_jy_{n+1-j}-z_jz_{n+1-j}\right) {\bf 1}\,. 
\end{eqnarray*}
With (\ref{W equation}) this yields
\[
G_{n+1}- \antcom{G_{n+1}}{\pi_0} = \sum_{j=1}^n \left( x_j x_{n+1-j} + y_jy_{n+1-j}-z_jz_{n+1-j}\right) Y\,.
\]
Next we compute with (\ref{derivatives}) that
\begin{equation}\label{pinprime}
\pi_n' = x_n' X + y_n' Y + z_n' Z - \theta' y_n Z - \theta' z_n Y\,,
\end{equation}
which implies with (\ref{commutators}) that
\[
[\pi_n',\pi_0] = x_n' Z + z_n'X - \theta'y_n X\,.
\]
In summary we obtain
\[
\pi_{n+1} =  \sum_{j=1}^n \left( x_j x_{n+1-j} + y_jy_{n+1-j}-z_jz_{n+1-j}\right) Y - \I \left(
x_n' Z + z_n'X - \theta'y_n X\right)\,,
\]
from which (\ref{xRec})--(\ref{zRec}) can be read off. While (\ref{E1a}) and (\ref{E1c}) follow directly from (\ref{xRec}) and (\ref{zRec}), for (\ref{E1b}) we observe that 
\[
0=\pi_0\pi_n'\pi_0 = (-y_n' + \theta' z_n)\pi_0\,,
\]
where the first equality is  (\ref{offdiag}) and for the second equality we use $\pi_0 X \pi_0 = \pi_0 Z \pi_0 =0$ (a consequence of (\ref{commutators})) and $\pi_0 Y\pi_0 = -\pi_0$ (a consequence of (\ref{W equation})) in (\ref{pinprime}).
 \qedi
\end{proof}

\begin{remark}
From  (\ref{E1a}) through (\ref{E1c}) we may derive recursions for
calculating $x_n$ or $z_n$, i.e.\
\begin{equation} \label{xn recursion with integration constant}
 x_{n+2}(\xi) = -   x''_n(\xi) +  \theta'(\xi)  \left( \int \theta'(\xi) x'_n(\xi) \, \D \xi + C \right) .
\end{equation}
and
\begin{equation} \label{zn recursion with integration constant}
 z_{n+2}(\xi) = - \frac{\D}{\D \xi} \left( z'_n(\xi) -  \theta'(\xi)  \left( \int \theta'(\xi) z_n(\xi) \, \D \xi + C \right) \right).
\end{equation}
The constant of integration $C$ must (and in some cases can) be
determined by comparison with (\ref{xRec})--(\ref{zRec}). In the case
where $\theta'$ is given as in Assumption 2, this strategy will
lead to fairly explicit expressions of the leading order of the coefficient functions
$x_n$, $y_n$ and $z_n$, cf.\ Proposition~\ref{z_n and recursion};
from these we will extract the asymptotic behaviour of $x_n$, $y_n$
and $z_n$.
\end{remark}

Using (\ref{E1a})--(\ref{E1c}), we can give very simple
expressions for the quantities appearing in (\ref{pi1}) and
(\ref{pi2}). As for (\ref{pi2}), we use (\ref{derivatives}) and
the differential equations to find
\begin{equation}
  \commut{\I\eps \partial_\xi - H}{\pi^{(n)}} = \I\eps^{n+1} \pi_n'
  = - \eps^{n+1} (z_{n+1} X + x_{n+1} Z) \label{pi_n'}.
 \end{equation}
The term by which
$\pi^{(n)}$ fails to be a projector is controlled in the following proposition.
\begin{proposition} \label{proj correction}
For each $n \in \N$, there exist functions $g_{n+1,k}$,  $k \leq n$,
with
\begin{equation} \label{projector error}
((\pi^{(n)})^2 - \pi^{(n)})(\xi) = \left(\sum_{k=1}^{n} \eps^{n+k}
g_{n+1,k}(\xi)\right) {\bf 1}.
\end{equation}
For each $k \leq n$,
$$g'_{n+1,k}  =  2 \I ( x_k  z _{n+1} -  z_k  x_{n+1}).$$
\end{proposition}
\begin{proof}
Recalling (\ref{ansatz}) we have that 
\[
(\pi^{(n)})^2 - \pi^{(n)} = \sum_{k=1}^{n} \eps^{n+k} \sum_{j=0}^{n-k} \pi_{k+j}\pi_{n-j}\,.
\]
As before we compute
\begin{eqnarray*}
 \sum_{j=0}^{n-k} \pi_{k+j}\pi_{n-j} &=& \sum_{j=0}^{n-k} \left(x_{k+j}X+ y_{k+j}Y+z_{k+j}Z  \right)\left( x_{n-j}X+ y_{n-j}Y+z_{n-j}Z  \right)\\&=&
 \sum_{j=0}^{n-k}  \left( x_{k+j} x_{n-j} + y_{k+j}y_{n-j}-z_{k+j}z_{n-j}\right) {\bf 1} =: g_{n+1,k} {\bf 1}\,. 
\end{eqnarray*}
  With Proposition \ref{function recursion},
\begin{eqnarray*}
g'_{n+1,k} & = & \sum_{j=0}^{n-k} \I (  z_{k+j+1}  x_{n-j} +
x_{k+j}  z_{n-j+1} ) +
( \theta'  z_{k+j}  y_{n-j} +  \theta'  y_{k+j}  z_{n-j}) + \\
&&  -\, \I (  x_{k+j+1}  z_{n-j} +  z_{k+j}  x_{n-j+1} )  - (\theta'  y_{k+j}  z_{n-j} +  \theta'  z_{k+j}  y_{n-j})= \\
& = & \I \sum_{j=0}^{n-k} (   z_{k+j+1}  x_{n-j} - z_{k+j}  x_{n-j+1}  +
x_{k+j}  z_{n-j+1} -  x_{k+j+1}  z_{n-j} )   = \\
& = & 2 \I ( x_k z_{n+1} -  z_k  x_{n+1} ).
\end{eqnarray*}
The last equality follows because the sum is a telescopic sum.
\qedi
\end{proof}

Since ${\mathbf 1} = \id$ is independent of $t$, Proposition \ref{proj
correction} gives the derivative of the correction $(\pi^{(n)})^2
- \pi^{(n)}$ to a projector. As above, this gives an easy way for
estimating the correction itself provided we have some clue how to
choose the constant of integration.

Our next task will be to construct, from the superadiabatic projections $\pi^{(n)}$, the transformation  $T_\eps^n$
into the $n^{\rm th}$ superadiabatic representation. 
By (\ref{ansatz}) and
(\ref{recursion1}), $\pi^{(n)}$ is almost a projection. Thus it has two eigenvectors $v_n$ and $w_n$. Let
$$v_0 = \left( \begin{array}{c} \frac{1}{\sqrt{p}}\\\sqrt{p} \end{array} \right), \quad
w_0 = \left( \begin{array}{c}  \frac{1}{\sqrt{p}}\\ -\sqrt{p} 
\end{array} \right)$$ be the eigenvectors of $\pi_0$, and write
\begin{equation} \label{v_n ansatz}
 v_n = \alpha v_0 + \beta w_0, \quad w_n = \overline{\alpha} w_0 - \overline{\beta} v_0 \qquad (\alpha, \beta \in \C).
 \end{equation}
We make this representation unique by requiring that $0 \leq \alpha \in
\R$ and $\alpha^2 +|\beta|^2= 1$, where the latter condition ensures that the determinant of $T_\eps^n$ is constant equal to $2$, which is the value for $T_0$. 

Let $T_\eps$ be the linear operator taking  $(v_n,w_n)$ to
the standard basis $(e_1,e_2)$ of $\R^2$ , i.e.
\begin{equation} \label{Tminus}
  T_\eps^{ -1} =  (v_n, w_n) =  \left( \begin{array}{cc} \frac{\alpha+\beta}{\sqrt{p}}
  & \frac{\alpha-\overline\beta}{\sqrt{p}} \\
  (\alpha-\beta) \sqrt{p} & - (\alpha+\overline \beta) \sqrt{p} \end{array} \right)
  \end{equation}
  and
  \begin{equation} \label{T}
  T_\eps  = \frac{1}{2} \left( \begin{array}{cc} (\alpha+  \overline\beta) \sqrt{p} 
  & \frac{\alpha-\overline\beta}{\sqrt{p}} \\
  (\alpha-\beta) \sqrt{p} & - \frac{\alpha+\beta}{\sqrt{p}}
 \end{array} \right)\,.
  \end{equation}
where all vectors are column vectors.
$T_\eps$ diagonalizes $\pi^{(n)}$, thus
\begin{equation} \label{diag pi_n}
  T_\eps \pi^{(n)} T_\eps^{ -1} = D \equiv \left( \begin{array}{cc} \lambda_1 & 0 \\
        0 & \lambda_2 \end{array} \right),
\end{equation}
where $\lambda_{1,2}$ are the eigenvalues of $\pi^{(n)}$. Although
$\alpha, \beta$ and $\lambda_{1,2}$ depend on $n$, $\eps$ and $\xi$,
we suppress this from the notation.

\begin{lemma} \label{T errors}
$$ T_0 T_\eps^{-1} = \left( \begin{array}{cc} \alpha & -\overline{\beta} \\
                                \beta & \alpha \end{array} \right),  \quad
    \mbox{and} \quad T_0 (T_\eps^{-1})' =
    \left( \begin{array}{cc}
    \alpha'  & -\overline\beta' \\
    \beta'  & \alpha'  
     \end{array} \right)-\frac{1}{2}\theta'
       \left( \begin{array}{cc}
    \beta &  \alpha\\
    \alpha & - \overline\beta
     \end{array} \right)
      $$
\end{lemma}
\begin{proof}
Just compute, using
  \begin{equation} \label{T0}
  T_0   = \frac{1}{2} \left( \begin{array}{cc}  \sqrt{p} 
  & \frac{1}{\sqrt{p}} \\
   \sqrt{p} & - \frac{1}{\sqrt{p}}
 \end{array} \right) 
  \end{equation}
  and 
   \[
  (T_\eps^{-1})' =  \left( \begin{array}{cc} \frac{\alpha'+\beta'}{\sqrt{p}}
  & \frac{\alpha'-\overline\beta'}{\sqrt{p}} \\
  (\alpha'-\beta') \sqrt{p} & - (\alpha'+\overline \beta') \sqrt{p} \end{array} \right) + \frac{p'}{2p} \left( \begin{array}{cc} -\frac{\alpha+\beta}{\sqrt{p}}
  & -\frac{\alpha-\overline\beta}{\sqrt{p}} \\
  (\alpha-\beta) \sqrt{p} & - (\alpha+\overline \beta) \sqrt{p} \end{array} \right)\,.
  \]
\end{proof}

It will turn out that $\beta, \alpha' \alpha$, and $\beta'$  are
small  quantities, $\lambda'_1, \lambda'_2$, and $\lambda_2$ are
even much smaller, while $\alpha^2$ and $\lambda_1$ are large,
i.e. of order $1 + \Or(\eps)$. This motivates the form in which we
present the following result.

\begin{proposition} \label{general transformed matrix}
Suppose $\lambda_1 \neq \lambda_2$. Then for each $n \in \N$,
$$ T_\eps  (\I \eps \partial_\xi - H) T_\eps^{-1} =  \I \eps \partial_\xi -
\left( \begin{array}{cc}  \frac{1}{2} &  \frac{\alpha^2 \eps^{n+1}}{\lambda_1 - \lambda_2} (x_{n+1} - z_{n+1}) \\
        \frac{\alpha^2 \eps^{n+1}}{\lambda_1 -\lambda_2} (x_{n+1} +z_{n+1}) & -\frac{1}{2} \end{array} \right) + R,$$
with
$$ R = {
\left( \begin{array}{cc}  |\beta|^2 - \eps (\tIm (\overline{\beta} \beta') + \I \theta'\alpha\tRe\beta) &
 \frac{\eps^{n+1} \overline{\beta}^2}{\lambda_1 - \lambda_2} (x_{n+1} + z_{n+1})  \\
        \frac{\eps^{n+1} \beta^2}{\lambda_1 -\lambda_2} (x_{n+1} - z_{n+1})  &
         -|\beta|^2 +  \eps (\tIm (\overline{\beta} \beta') + \I \theta'\alpha\tRe\beta) \end{array} \right)}.$$
\end{proposition}
\begin{proof}
Let us write $T_\eps  (\I \eps \partial_\xi - H) T_\eps^{-1} =
(M_{i,j})$, $i,j \in \{1,2\}$. $M_{1,1}$ and $M_{2,2}$ are
calculated in a straightforward manner, using Lemma \ref{T errors}
together with the
 fact $T_0 H T_0^{-1} ={\rm diag}(\frac{1}{2},-\frac{1}{2})$:
\begin{eqnarray*}
 T_\eps  (\I \eps \partial_\xi - H) T_\eps^{-1}  &= & \I
\eps
\partial_\xi + \I \eps T_\eps  T_0^{-1} T_0 (T_\eps^{-1})' -
T_\eps  T_0^{-1} T_0 H T_0^{-1} T_0 T_\eps^{-1} = \\
&=& \I \eps \partial_\xi + \I \eps
    \left( \begin{array}{cc} \alpha & \overline{\beta} \\
                                -\beta & \alpha \end{array} \right)
    \left( \begin{array}{cc}
    \alpha' -\frac{1}{2}\theta'\beta & -\overline\beta'-\frac{1}{2}\theta'\alpha \\
    \beta' -\frac{1}{2}\theta' \alpha& \alpha'  +\frac{1}{2}\theta'\overline\beta
     \end{array} \right) 
                  \\ & & \, -
    \frac{1}{2} \left( \begin{array}{cc} \alpha & \overline{\beta} \\
                -\beta & \alpha \end{array} \right)
    \left( \begin{array}{cc} 1 & 0 \\
        0 & -1 \end{array} \right)
    \left( \begin{array}{cc} \alpha & -\overline{\beta} \\
                                \beta & \alpha \end{array} \right).
\end{eqnarray*}
Carrying out the matrix multiplication yields
\begin{eqnarray*}
 M_{1,1} &=& \I \eps \partial_\xi + \I \eps ( \alpha\alpha' +\overline\beta\beta'-\theta'\alpha\tRe\beta) - {\textstyle\frac{1}{2} }(\alpha^2 - |\beta|^2) \\  
 M_{2,2} &=& \I \eps \partial_\xi + \I \eps ( \alpha\alpha' +\beta\overline\beta'+\theta'\alpha\tRe\beta )+ {\textstyle\frac{1}{2} }(\alpha^2 - |\beta|^2)\,.
 \end{eqnarray*}
This gives the diagonal coefficients of $M$ when we use  
$\alpha^2 + |\beta|^2 = 1$, so that $\alpha^2
-|\beta|^2 = 1 - 2 |\beta|^2$ and $\alpha \alpha' = - \frac{1}{2} (\beta \overline{\beta}' + \overline{\beta} \beta')$.   
Although we could get expressions for the off-diagonal
coefficients by the same method, these would not be useful later
on. Instead we use (\ref{diag pi_n}), i.e. $T_\eps^{-1} D =
\pi^{(n)} T_\eps^{-1}$ together with (\ref{pi_n'}) and obtain
\begin{equation} \label{commut trick}
 T_\eps  (\I \eps \partial_\xi - H) T_\eps^{-1}D = DT_\eps  (\I \eps \partial_\xi - H)
 T_\eps^{-1} - \eps^{n+1} T_\eps  (z_{n+1} X + x_{n+1} Z) T_\eps^{-1}.
\end{equation}
By multiplying (\ref{commut trick}) with $e_j e_j^{\ast}$ from the
left and by $e_k e_k^{\ast}$ from the right, $j,k \in \{1,2\}$, $j\not=0$,
and rearranging, we obtain
\begin{eqnarray} \label{commut trick 2}
\lefteqn{(\lambda_k - \lambda_j) \,e_j \,e_j^{\ast} \,T_\eps  (\I
\eps
\partial_\xi - H)\,
T_\eps^{-1} e_k \,e_k^{\ast} = } \\
& = & - \eps^{n+1}
 e_j \,e_j^{\ast}  \,T_\eps  (z_{n+1} X + x_{n+1} Z) \,T_\eps^{-1} e_k\,
 e_k^{\ast} \,. \nonumber
 \end{eqnarray}
From the equalities
$ T_0 X T_0^{-1} = { \left( \begin{array}{cc} 0 & -1 \\
                    -1 & 0 \end{array} \right)}$, 
                    $T_0 Z T_0^{-1} =
                    { \left( \begin{array}{cc} 0 & 1 \\
                    -1 & 0 \end{array} \right)}$
and Lemma~\ref{T errors} we obtain
$$
        T_\eps  X T_\eps^{-1} = {
        \left( \begin{array}{cc}
        -\alpha(\beta+\overline\beta) & -(\alpha^2 - \overline\beta^2) \\
        -(\alpha^2-\beta^2) & \alpha(\beta + \overline{\beta}) \end{array} \right)}.$$
        $$ T_\eps Z T_\eps^{-1} =
\left( \begin{array}{cc}
\alpha(\beta-\overline\beta) & \alpha^2 + \overline\beta^2 \\
        -(\alpha^2+\beta^2) & -\alpha(\beta - \overline{\beta}) \end{array}
        \right)\,,
$$
The expressions for $M_{1,2}$ and $M_{2,1}$ now follow from (\ref{commut trick 2}). \qedi \end{proof}

We now use our results from the previous section to express
$\alpha, \beta$ and $\lambda_{1,2}$ in terms of $x_k, y_k$ and
$z_k$, $k \leq n$. Let us define
\begin{eqnarray}
\label{chi} \chi & \equiv & \chi(n,\eps,\xi) = {\textstyle \sum_{k=1}^n} \eps^k x_k(\xi), \\
\label{eta} \eta & \equiv & \eta(n,\eps,\xi) = {\textstyle\sum_{k=1}^n} \eps^k y_k(\xi), \\
\label{zeta} \zeta & \equiv & \zeta(n,\eps,\xi) =
{\textstyle\sum_{k=1}^n} \eps^k z_k(\xi).
\end{eqnarray}
Moreover, let
\begin{equation} \label{g}
 g \equiv g(n,\eps,\xi) = {\textstyle\sum_{k=1}^{n}} \eps^{n+k} g_{n+1,k}(\xi)
\end{equation}
be the quantity appearing in (\ref{projector error}).

\begin{lemma} \label{eigenvalues}
The eigenvalues of $\pi^{(n)}$ solve the quadratic equation
$$ \lambda_{1,2}^2 - \lambda_{1,2} - g = 0.$$
\end{lemma}
\begin{proof}
By (\ref{diag pi_n}) and Proposition \ref{proj correction} we
obtain
$$ \left(\begin{array}{cc} \lambda_1^2 - \lambda_1 & 0 \\ 0 & \lambda_2^2 -
\lambda_2 \end{array} \right) = T_\eps  ((\pi^{(n)})^2 -
\pi^{(n)}) \,T_\eps^{-1} = T_\eps  g{\bf 1} T_\eps^{-1} =
\left(
\begin{array}{cc} \,g\, & \,0\, \\[2pt] \,0\, & \,g\, \end{array} \right)\,.\quad\qedi $$
\end{proof}

\begin{lemma} \label{alpha beta}
$$ \alpha^2 (\lambda_1 - \lambda_2) = 1 - \eta - \lambda_2, \qquad \mbox{and }
\qquad \alpha \beta (\lambda_1 - \lambda_2)  = - \chi - \zeta.$$
\end{lemma}
\begin{proof}
From (\ref{v_n ansatz}) and $\alpha^2+|\beta|^2=1$ one obtains
\begin{equation} \label{no1}
v_0 = \alpha v_n -\beta w_n \,.
\end{equation}
Thus, using again  (\ref{v_n ansatz}), we find
\begin{eqnarray*}
\pi^{(n)} v_0 &=& \lambda_1 \alpha v_n - \lambda_2 \beta w_n = (\lambda_1\alpha^2 +
\lambda_2 |\beta|^2) v_0 + (\lambda_1 - \lambda_2)\alpha \beta w_0 = \\
&=& (\alpha^2(\lambda_1 - \lambda_2) + \lambda_2) v_0 + (\lambda_1
- \lambda_2)\alpha \beta w_0.
\end{eqnarray*}
  On the other
hand, from (\ref{ansatz}) and (\ref{recursion1}) we have
\begin{equation}\label{p0minuspi}
\pi^{(n)} = \pi_0 + \sum_{k=1}^n \eps^k (x_k X + y_kY + z_kZ)\,,
\end{equation}
and since $Xv_0 = Zv_0 = -w_0$, $\pi_0 v_0 = v_0$ and $Y v_0 = -
v_0$, we find
$$ \pi^{(n)} v_0 = (1 - \eta) v_0 - (\chi + \zeta) w_0.$$
Comparing coefficients finishes the proof. \qedi \end{proof}

We can now prove Theorem \ref{n-th}, which we restate here in slightly different form. To make the connection, note that $x_n$ is purely imaginary 
and $z_n$ is purely real, so that indeed $x_n + z_n = - \overline x_n + \overline z_n$.

\begin{theorem} \label{general diag} Let $\eps_0>0$ be
sufficiently small. For $\eps\in(0,\eps_0]$ assume there is a
bounded function $q$ on $\R$ such that $\chi(\xi)$, $\eta(\xi)$,
$\zeta(\xi)$ and their derivatives $\chi'(\xi), \eta'(\xi), \zeta'(\xi)$
are all bounded in norm by $\eps q(\xi)$. Then
\begin{eqnarray} \label{general offdiag eq}
 &&\hspace{-21pt} T_\eps 
 (\I \eps \partial_\xi - H) T_\eps^{-1} =\hfill\\&&\hspace{-21pt}= \I \eps \partial_\xi-
\left( \begin{array}{cc} \frac{1}{2} + \Or(\eps^2 q) & -\eps^{n+1}
(x_{n+1} - z_{n+1})\, (1 + \Or(\eps q))
\\[2mm]
\eps^{n+1} (x_{n+1} + z_{n+1})\, (1 + \Or(\eps q)) & -\frac{1}{2}
+ \Or(\eps^2 q) \end{array} \right)\,.\nonumber
\end{eqnarray}
\end{theorem}
\begin{proof}
From (\ref{p0minuspi}) and our assumptions it follows that
$\pi^{(n)} - \pi_0 = \Or(\eps q)$. Thus $\lambda_1 = 1 + \Or(\eps
q)$ and $\lambda_2 = \Or(\eps q)$, and from Lemma \ref{eigenvalues}
we infer $g = \Or(\eps q)$ and
$$ \lambda_1 = {\textstyle\frac{1}{2}} \left( 1 + \sqrt{1 + 4 g} \right)\,,\qquad\lambda_2
= {\textstyle\frac{1}{2}} \left( 1 - \sqrt{1 + 4 g} \right)\,.$$
Since $\lambda_1 - \lambda_2 \neq 0$, Lemma \ref{alpha beta}
yields
$$ \alpha^2 = \frac{1 + \sqrt{1+4g} - 2 \eta}{2 \sqrt{1 + 4g}}, \quad \beta =
\frac{-\xi - \zeta}{\sqrt{1+4g} \alpha}\,.$$ Hence $ \alpha^2 = 1
+ \Or(\eps q),$ and  $\beta, \beta' $ and $ \alpha \alpha' =
(\alpha^2)' /2$ are all $\Or(\eps q)$. We now plug these into the
matrix $R$ in Proposition \ref{general transformed matrix}. This shows
the claim. \qedi
\end{proof}

\section{Asymptotic behaviour of superadiabatic representations} \label{asympt recursion}

The next step in the program is to 
understand the asymptotic behavior of the off-diagonal elements of the effective
Hamiltonian in the $n^{\rm th}$ superadiabatic basis for large
$n$. According to (\ref{general offdiag eq}) this amounts to the
asymptotics of $x_n$ and $z_n$ as given by the recursion from
Proposition \ref{function recursion}. It is clear that the
function $\theta'$ alone determines the behavior of this
recursion. As in \cite{BeTe1,BeTe2} it will in fact be the poles of this function closest to the real axis that play the dominant role; thus we first investigate the special case 
\begin{equation} \label{theta ansatz}
\theta'(\xi) =\gamma\left( \frac{ 1}{\xi+\I \xi_{\rm c}} + \frac{ 1}{\xi-\I \xi_{\rm c}} \right)\,,
\end{equation}
and then turn to a perturbation as allowed in Assumption 2. Fortunately, all the changes relative to \cite{BeTe1, BeTe2} occur in the error terms of the explicit part stemming from (\ref{theta ansatz}), and so we have to redo only a small part of the work.

We use (\ref{zn recursion with integration constant}) in order to
determine the asymptotics of $z_n$. From Proposition \ref{function
recursion} together with (\ref{xRec})--(\ref{zRec}) it is clear that
$y_n$ must go to zero as $\xi \to \pm \infty$. This fixes the
constant of integration in (\ref{zn recursion with integration
constant}), and we arrive at the linear two-step recursion
\begin{equation} \label{zn recursion}
 z_{n+2}(\xi) = - \frac{\D}{\D \xi} \left( z'_n(\xi) -  \theta'(\xi) \int_{-\infty}^\xi \theta'(s) z_n(s) \, \D s  \right).
\end{equation}

Let us define
\[
 \theta' = \frac{-\I \gamma}{\xi_{\rm c}} (f - \overline f) \qquad \mbox{with}\qquad f(\xi)
= \frac{ \I\xi_{\rm c}}{\xi + \I \xi_{\rm c}}\,.
\]
and let us put $\sigma_j = 1 - 2 \delta_{0,j}$, i.e.\ $\sigma_j=-1$ for $j=0$ and $\sigma_j=1$ otherwise. 
We obtain
\begin{lemma} \label{identities}
For each $m \geq 1$,
\begin{eqnarray}
\theta' \tRe(f^m) &=&  \frac{\gamma}{\xi_{\rm c}} \tIm(f^{m+1}) - \frac{\gamma}{\xi_{\rm c}} 
\sum_{k=0}^{m-2}\left( \frac{1}{2}\right)^{k+1} \tIm \left(f^{m-k} \right)\\ 
&=& - \frac{\gamma}{\xi_{\rm c}} \sum_{j=0}^{m-1} \sigma_j 2^{-j} \tIm \left(f^{m+1-j}\right) \label{Re mult} \\
\tIm(f^m)' &=& -\frac{m}{\xi_{\rm c}} \,\tRe(f^{m+1}), \label{Im diff}\\
\tRe(f^m)' &=& \frac{m}{\xi_{\rm c}}\, \tIm(f^{m+1}). \label{Re diff}
\end{eqnarray}
\end{lemma}
\begin{proof}
We have $  f\overline{f}= \frac{1}{2} (f + \overline{f}) $, and thus
\[
 f^k \overline f = \frac{1}{2} f^{k-1}(f+\overline f) = \frac{1}{2} (f^k + f^{k-1}\overline f)
 \]
and
\begin{equation} \label{multiplication rule}
\theta' f^{m} = \frac{-\I\gamma}{\xi_{\rm c}}\left(f^{m+1} - f^{m}\overline f \right)=
\frac{-\I\gamma}{\xi_{\rm c}}  \left( f^{m+1} -\sum_{k=0}^{m-1}
\left(\frac{1}{2}\right)^{k+1} f^{m-k}  - \left(\frac{1}{2}\right)^{m}\overline f \right).
\end{equation}
Taking the the  
real  part  of (\ref{multiplication rule}) yields
 (\ref{Re mult}). To prove
(\ref{Im diff}) and 
(\ref{Re diff}), it suffices to observe that
$(f^k)' =  \frac{k}{\I\xi_{\rm c}} f^{k+1} )$. \qedi
\end{proof}

\begin{proposition} \label{z_n and recursion}
For each even $n \in \N$ and $j = 0, \ldots, n-1$, let the numbers
$a_j^{(n)}$ be recursively defined through
\begin{eqnarray}
a_0^{(2)} & = & 1, \qquad a_1^{(2)} = 0\,, \label{a_n rekursion start}\\
a_j^{(n+2)} &=& \frac{n+1-j}{(n+1)\,n} \left( (n-j)\,a_j^{(n)} -
\gamma^2 \sum_{k=0}^j \frac{\sigma_{j-k}}{n-k}
\sum_{m=0}^k \sigma_{k-m} a_m^{(n)} \right) 
 \label{a_n rekursion step} 
\end{eqnarray}
for $j \leq n$ and $a_{n+1}^{(n+2)} = 0\,.$
Then
\begin{eqnarray}
 z_n & = & - \gamma \frac{(n-1)!}{\xi_{\rm c}^n} \sum_{j=0}^{n-1} 2^{-j} a_j^{(n)}
 \tRe(f^{n-j}) \quad (n \mbox{ even})\,,
  \label{z_n recursion step} \\
 y_n & = &  - \gamma^2 \frac{(n-1)!}{\xi_{\rm c}^{n}} \sum_{j=0}^{n-1} 2^{-j} \left(
 \frac{1}{n-j} \sum_{k=0}^j \sigma_{j-k}a_k^{(n)}\right) \tRe(f^{n-j})
                    \quad (n \mbox{ even})\,,  \label{y_n recursion step} \\
 x_n & = &  \I \gamma \frac{(n-1)!}{t_{\rm c}^{n}} \sum_{j=0}^{n-1} 2^{-j} \left(
 \frac{n}{n-j} a_j^{(n+1)}
 \right) \tIm(f^{n-j}) \quad (n \mbox{ odd})\,,  \label{x_n recursion step}
\end{eqnarray}
\end{proposition}
\begin{proof}
We proceed by induction. We have $  x_1 = \I  \theta' / 2 =
\frac{\I \gamma}{\xi_{\rm c}} \tIm(f)$, and thus by (\ref{E1a}) and
(\ref{Re diff}),
$$   z_2 =  - \I x'_1 = - \frac{\gamma}{\xi_{\rm c}^2} \tRe(f^2).$$
This proves (\ref{z_n recursion step}) for $n=2$. Now suppose that
(\ref{z_n recursion step}) holds for some even  $n \in \N$. Then
by (\ref{E1a}) and (\ref{Re diff}), (\ref{x_n recursion step})
holds for $n-1$. To prove (\ref{y_n recursion step}) for the given
$n$, we want to use (\ref{E1b}). (\ref{Re mult}) and the induction
hypothesis on $z_n$ yield
\begin{eqnarray}
  \theta'   z_n & = & - \gamma^2 \frac{(n-1)!}{\xi_{\rm c}^{n+1}} \sum_{j=0}^{n-1}a_j^{(n)}
  \sum_{k=0}^{n-j-1} \sigma_k 2^{-(k+j)} \tIm(f^{n+1-(j+k)})  = \nonumber \\
& = &  - \gamma^2 \frac{(n-1)!}{\xi_{\rm c}^{n+1}} \sum_{m=0}^{n-1}
2^{-m} \left(\sum_{j=0}^m \sigma_{m-j} a_j^{(n)}\right) \tIm(f^{n+1-m})
\label{E5}.
\end{eqnarray}
Since (\ref{E5}) only contains second or higher order powers of
$f$, it is easy to integrate using (\ref{Re diff}). Let us write
\begin{equation} \label{bj in proof}
b_m^{(n)} = \frac{1}{n-m} \sum_{j=0}^m \sigma_{m-j} a_j^{(n)}.
\end{equation}
Then by (\ref{Re diff}) we obtain
$$  y_n =  \int_{-\infty}^\xi  \theta'(s)   z_n(s) \, \D s =  -\gamma^2 \frac{(n-1)!}{\xi_{\rm c}^{n-1}}
\sum_{m=0}^{n-1} 2^{-m} b_m^{(n)} \tRe(f^{n-m}),$$ proving
(\ref{y_n recursion step}) for $n$. It remains to prove (\ref{z_n
recursion step}) for $n+2$. We want to
 use (\ref{zn recursion}), and therefore we employ (\ref{Re mult}) and our above calculations in order to get
\begin{eqnarray*}
 \lefteqn{ \theta'(t ) \int_{-\infty}^t    \theta'(s)   z_n(s) \, \D s
 = }\\
 &=& \hspace{-6pt}  \gamma^3 \frac{(n-1)!}{\xi_{\rm c}^{n+1}}
\sum_{j=0}^{n-1} b_j \sum_{k=0}^{n-j+1} 2^{-(k+j)} \sigma_k \tIm(f^{n+1-(k+j)})  =   \\
& =& \hspace{-6pt}  \gamma^3 \frac{(n-1)!}{\xi_{\rm c}^{n+1}}  \sum_{j=0}^{n-1} 2^{-j} \left( \sum_{k=0}^j \sigma_{j-k} b_k \right)
\tIm(f^{n+1-j}).
\end{eqnarray*}
By (\ref{Im diff}),
$$    z'_n = - \gamma \frac{(n-1)!}{t_{\rm c}^{n+1}} \sum_{j=0}^{n-1} 2^{-j} a_j^{(n)} (n-j) \tIm(f^{n+1-j}).$$
Now we take the difference of the last two expressions, take the derivative of that and obtain
\[
  z_{n+2} = - \gamma \frac{(n-1)!}{t_{\rm c}^{n+2}}  \sum_{j=0}^{n-1} 2^{-j} (n+1-j)
   \left( (n-j) a_j^{(n)} - \gamma^2 \sum_{k=0}^j \sigma_{j-k} b_k \right) \tRe(f^{n+2-j}).
\]
Comparing coefficients, this proves (\ref{z_n recursion step}) for
$n+2$. \qedi
\end{proof}

We now investigate the behavior of the coefficients $a_j^{(n)}$ as
$n \to \infty$.

\begin{proposition} \label{a_n limit}
Let $a_j^{(n)}$ be defined as in Proposition~\ref{z_n and
recursion}. \alp
\begin{enumerate}
\item  $a_0^{(n)} = {\displaystyle \frac{\sin (\gamma\pi /
2)}{\gamma \pi/2} }\left(1 +
\Or\left(\frac{\gamma^2}{n^2}\right)\right).$
\item There exists $C_1 > 0$ such that for all $n \in \N$
$$ |a_1^{(n)}| \leq C_1 \frac{\ln n}{n-1}\,.$$
 \item For each $p>1$ there exists $C_2 > 0$
such that for all $n \in \N$
$$ \sup_{j \geq 2} p^{-j} |a_j^{(n)}| \leq \frac{C_2}{n-1}\,.$$

\end{enumerate}
\end{proposition}

\begin{proof}
(a) By (\ref{a_n rekursion start}), $a_0^{(2)} = 1$, and
$$ a_0^{(n+2)}  = a_0^{(n)}\left( 1 - \frac{\gamma^2}{n^2} \right).$$
Comparing with the product representation of the sine function (\cite{AbSt}, 4.3.89)
$$ \sin(\pi x) = \pi x \prod_{n=1}^{\infty} \left( 1 - \frac{x^2}{n^2} \right)\,,$$
we arrive at (a). \\
(b) Put $\alpha_n = (n-1) a_1^{(n)}$. Then by (\ref{a_n rekursion
step}),
$$ \alpha_{n+2} = \alpha_n \left( 1 - \frac{\sigma_0^2\gamma^2}{(n-1)^2} \right) - \gamma^2 \sigma_0\sigma_1
\left( \frac{1}{n} + \frac{1}{n-1} \right) a_0^{(n)}.$$ thus for
$n-1 > \gamma$, we have
$$ |\alpha_{n+2}| \leq |\alpha_n| + \gamma^2 \left( \frac{1}{n} + \frac{1}{n-1} \right)
\max_{m \in \N} |a_0^{(m)}|,$$
which shows (b). \\
(c) Put $c_j^{(n)} = (n-1) p^{-j} a_j^{(n)}$, and $c^{(n)} =
\max_{j \geq 2} |c_j^{(n)}|$. We will show that the sequence
$c^{(n)}$ is bounded. We have
\begin{eqnarray}
c_j^{(n+2)} & = & \frac{n+1-j}{n(n-1)} \left( (n-j) c_j^{(n)} - \gamma^2 \sum_{k=2}^j \frac{\sigma_{j-k}}{n-k}
\sum_{m=2}^k \sigma_{k-m}p^{-j+m} c_m^{(n)} - \right. \nonumber \\
&& \left. -(n-1) p^{-j} \gamma^2 \left(  a_0^{(n)} \sum_{k=0}^j
\frac{\sigma_k\sigma_{j-k}}{n-k} +  a_1^{(n)} \sum_{k=1}^j \frac{\sigma_{k-1} \sigma_1}{n-k} \right)
\right). \label{bigformula}
\end{eqnarray}
Now
\begin{equation} \label{term1}
 \left| \sum_{k=2}^j \frac{\sigma_{j-k}}{n-k} \sum_{m=2}^k 
 \sigma_{k-m}p^{-j+m} c_m^{(n)} \right|
\leq c^{(n)} \frac{1}{n-j}\frac{p^2}{(p-1)^2},
\end{equation}
and
\begin{equation} \label{term2}
 p^{-j} \sum_{k=0}^j \frac{1}{n-k} \leq \frac{(j+1) p^{-j}}{n-j} \leq \frac{1}{(n-j)\ln p}.
\end{equation}
We plug these results into (\ref{bigformula}) and obtain
\begin{eqnarray*} 
|c_j^{(n+2)}| &\leq& c^{(n)} \left( \frac{(n+1-j)(n-j)}{n(n-1)}  +
\frac{(n+1-j)\,\gamma^2 p^2}{(n-j)(p-1)^2} \frac{1}{n(n-1)} \right) +\\
& & + \frac{1}{n} \frac{(n+1-j)\,p^2\,\gamma^2}{(n-j) (p-1)^2 \ln
p}(|a_0^{(n)}|+|a_1^{(n)}|).
\end{eqnarray*}
By (a) and (b),  $a_0^{(n)}$ and $a_1^{(n)}$ are bounded. Taking
the supremum over $j \geq 2$ above, we see that  there exist
constants $B_1$ and $B_2$ with
$$ c^{(n+2)} \leq c^{(n)} \left( \frac{n-2}{n} + \frac{B_1}{n(n-1)} \right) + \frac{B_2}{n},$$
hence
$$ c^{(n+2)} - c^{(n)} \leq \frac{1}{n} \left(\left( -2 + \frac{B_1}{n-1}\right) c^{(n)} + B_2 \right).$$
Now let $n-1 > B_1$. Then for $c^{(n)} > B_2$, the above
inequality shows $c^{(n+2)} < c^{(n)}$, while for $c^{(n)} \leq
B_2$, $c^{(n+2)} \leq c^{(n)} + B_2/n \leq B_2 (1 + 1/n).$ Thus
$c^{(n)}$ is a bounded sequence. \qedi
\end{proof}

\begin{corollary} \label{b_n limit}
Let $b_j^{(n)}$ be given by (\ref{bj in proof}). Then for each $p
> 1$, there exists $C_3>0$ such that
$$ \sup_{j \geq 0} p^{-j} b_j^{(n)} \leq \frac{C_3}{n-1}.$$
\end{corollary}

\begin{proof}
For $j \leq n-1$, we have $n-1 \leq j(n-j)$, and thus
Proposition~\ref{a_n limit} (c) gives
\begin{eqnarray*}
p^{-j} b_j^{(n)} & \leq & \frac{p^{-j}}{n-j} \left( (a_0^{(n)} + a_1^{(n)})
+ \frac{C_2}{n-1} \frac{p^2}{p-1} (p^{j-1}-1) \right) \leq \\
& \leq & p^{-j}  \frac{j(a_0^{(n)} + a_1^{(n)})}{n-1} + \frac{p\,
C_2}{p-1} \frac{1}{n-1} \leq \frac{C_3}{n-1}\,.\qquad\qedi
\end{eqnarray*}
\end{proof}

Given the results of Proposition \ref{a_n limit} and Corollary \ref{b_n limit} together with Lemma \ref{identities}, we can now derive 
Theorems \ref{MainTh} and \ref{solution}. The proofs are word for word the same as in \cite{BeTe1} and \cite{BeTe2} from here on, and so we omit them.


\begin{thebibliography}{99}


\bibitem[AbSt]{AbSt} M.\ Abramowitz and I.A.\ Stegun (Eds.). {\em Handbook of Mathematical Functions}, 9th printing, Dover, New York,
1972.

\bibitem[Be1]{Be1}  M.\ V.\ Berry. {\em Waves near Stokes Lines} 
Proc.\ R.\ Soc.\ Lond.\ A {\bf 427}, 265--280 (1990).

\bibitem[Be2]{Be2}  M.\ V.\ Berry. {\em Histories of adiabatic quantum transitions},
Proc.\ R.\ Soc.\ Lond.\ A {\bf 429}, 61--72 (1990).

\bibitem[BeLi]{BeLi}  M.\ V.\ Berry and R.\ Lim. {\em Universal transition prefactors
derived by superadiabatic renormalization}, J.\ Phys.\ A {\bf 26},
4737--4747 (1993).

\bibitem[BeTe1]{BeTe1}  V.\ Betz and S.\ Teufel {\em Precise coupling terms in
adiabatic quantum evolution},  
Annales Henri Poincare {\bf 6},  217-246 (2005).
\bibitem[BeTe2]{BeTe2}  V.\ Betz and S.\ Teufel {\em Precise coupling terms in adiabatic quantum evolution: the generic case}, Commun.\ Math.\ Phys.\ {\bf 260}, 481-509 (2005).
\bibitem[BeTe3]{BeTe3} V.\ Betz, S.\ Teufel, {\em Landau-Zener formulae from adiabatic transition histories}, in {\em Mathematical Physiucs of Quantum Mechanics,} Lecture Notes in Physics {\bf 690}, p.\ 19-32, Springer (2006).

\bibitem[Fe1]{Fe0} M.V.\ Fedoryuk, { M\'ethodes asymptotiques pour les \'equations diff\'erentielles ordinaires lin\'eaires}, Mir, Moscou (1987)

\bibitem[Fe2]{Fe} M.V.\ Fedoryuk, {\em Analysis I}, in 'Encyclopaedia of mathematical Sciences', volume 
13, R.V.\ Gamkrelidze Edt., Springer (1989).

\bibitem[FrFr]{FrFr} Froman, N., Froman, P.O.: JWKB approximation. Amsterdam: North-Holland (1965).

\bibitem[HaJo1]{hj} G.\ Hagedorn and A.\ Joye. {\em A Time-Dependent Born-Oppenheimer Approximation with Exponentially Small Error Estimates},
Commun.\ Math.\ Phys. {\bf  223}, 583-626 (2001).

\bibitem[HaJo2]{hj04} G.\ Hagedorn and A.\ Joye. {\em Time development of
exponentially small non-adiabatic transitions}, 
Commun.\ Math.\ Phys.\ {\bf 250}, 393-413 (2004).

\bibitem[HaJo3]{hj05} G.\ Hagedorn, A.\ Joye, {\em Determination of Non-Adiabatic Scattering Wave 
Functions in a Born-Oppenheimer Model}, Annales Henri Poincar\'e {\bf 6}, 937-990 (2005). 
 Erratum, ibid. p.1197-1199.

\bibitem[Jo]{Jo} A.\ Joye. {\em Exponential Asymptotics in a Singular Limit for $n$-Level Scattering Systems}, SIAM Journal on Math. Anal., {\bf 28} 669-703 (1997).


\bibitem[JKP]{JKP} A.\ Joye, H.\ Kunz and C.-E.\ Pfister. {\em
Exponential decay and geometric aspect of transition probabilities
in the adiabatic limit}, Ann.\ Phys.\ {\bf 208}, 299 (1991).

\bibitem[JoMa]{jm05} A.\ Joye, M.\ Marx,  {\em Semi-classical Determination of Exponentially Small 
Intermode Transitions for $1+1$ Space-time Scattering Systems}, Communications 
on Pure and Applied Mathematics {\bf 60}, 1189-1237 (2007).

\bibitem[JoPf$_1$]{JoPf1} A.\ Joye and C.-E.\ Pfister. {\em
Exponentially small adiabatic invariant for the Schr\"odinger
equation},
 Commun.\ Math.\ Phys.\ {\bf 140}, 15--41 (1991).

\bibitem[JoPf$_2$] {jp95} A.\ Joye, C.-E.\ Pfister, {\em Semiclassical Asymptotics Beyond all Orders for 
Simple Scattering Systems}, 
SIAM Journal on Mathematical Analysis {\bf 26}, 944-977 (1995).

\bibitem[LaLi]{LaLi} L.\ Landau, E.\ Lifschitz. {\em Quantum Mechanics}, Elsevier 1977.


\bibitem[Ma]{mz}  A.\ Martinez.,  An Introduction to Semiclassical and Microlocal Analysis, Springer UTX Series (2002)


\bibitem[Ne]{Ne}  G.\ Nenciu. \emph{Linear adiabatic theory.
Exponential estimates}, Commun.\ Math.\ Phys.\ \textbf{152},
479--496 (1993).


\bibitem[Ra]{Ra} T.\ Ramond, {\em Semiclassical study of quantum scattering on the line}, 
Communications in Mathematical Physics {\bf 177}, 221-254  (1996).

\bibitem[Sj]{Sj} J.\ Sj\"ostrand. {\em Projecteurs adiabatiques du point de vue pseudodiff\'erentiel}, C.\ R.\ Acad.\ Sci.\ Paris S\'er.\ I Math.\ {\bf 317}, 217-220 (1993).


\bibitem[WiMo]{WiMo} M.\ Wilkinson and M.\ Morgan. {\em
Nonadiabatic transitions in multilevel systems}, Phys.\ Rev.\ A
{\bf 61}, 062104 (2000).

\end{thebibliography}
\end{document}